\documentclass[aps,twocolumn,superscriptaddress]{revtex4-1}
\usepackage{amsmath,bbm}%,mathrsfs}
\usepackage{graphicx}
\usepackage{amsfonts}
\usepackage{amssymb}
\usepackage{amsmath, amssymb, amsthm, verbatim, changes}
\usepackage{color,xcolor,longtable}
\usepackage{slashbox}

\definecolor{myurlcolor}{rgb}{0,0,0.7}
\definecolor{myrefcolor}{rgb}{0.8,0,0}

\usepackage[unicode=true,pdfusetitle, bookmarks=true,bookmarksnumbered=false,bookmarksopen=false, breaklinks=false,pdfborder={0 0 0},backref=false,colorlinks=true, linkcolor=myrefcolor,citecolor=myurlcolor,urlcolor=myurlcolor]{hyperref}

\makeatletter
\newtheorem*{rep@theorem}{\rep@title}
\newcommand{\newreptheorem}[2]{%
\newenvironment{rep#1}[1]{%
 \def\rep@title{#2 \ref{##1}}%
 \begin{rep@theorem}}%
 {\end{rep@theorem}}}
\makeatother

\theoremstyle{plain}
\newtheorem{thm}{Theorem}%[section]
\newreptheorem{thm}{Theorem}
\newtheorem{lem}{Lemma}

\theoremstyle{definition}

\theoremstyle{remark}

\newcommand{\ket}[1]{|#1\rangle}
\newcommand{\bra}[1]{\langle#1|}
\newcommand{\proj}[1]{\ket{#1}\!\bra{#1}}
\newcommand{\ot}[0]{\otimes}

\begin{document}

%\preprint{}

\title{Bell inequalities tailored to maximally entangled states}


\author{Alexia Salavrakos}
\affiliation{ICFO-Institut de Ciencies Fotoniques, The Barcelona Institute of Science and Technology, 08860 Castelldefels, Barcelona, Spain}
\author{Remigiusz Augusiak}
\affiliation{Center for Theoretical Physics, Polish Academy of Sciences, Aleja Lotnik\'ow 32/46, 02-668 Warsaw, Poland}
\author{Jordi Tura}
\affiliation{ICFO-Institut de Ciencies Fotoniques, The Barcelona Institute of Science and Technology, 08860 Castelldefels, Barcelona, Spain}
\affiliation{Max-Planck-Institut f\"ur Quantenoptik, Hans-Kopfermann-Stra{\ss}e 1, 85748 Garching, Germany}
\author{Peter Wittek}
\affiliation{ICFO-Institut de Ciencies Fotoniques, The Barcelona Institute of Science and Technology, 08860 Castelldefels, Barcelona, Spain}
\affiliation{University of Bor{\aa}s, Allegatan 1, 50190 Bor{\aa}s, Sweden}
\author{Antonio Ac\'in}
\affiliation{ICFO-Institut de Ciencies Fotoniques, The Barcelona Institute of Science and Technology, 08860 Castelldefels, Barcelona, Spain}
\affiliation{ICREA - Institucio Catalana de Recerca i Estudis Avan\c{c}ats, E-08010 Barcelona, Spain}
\author{Stefano Pironio}
\affiliation{Laboratoire d'Information Quantique, CP 224, Universit\'e libre de Bruxelles (ULB), 1050 Bruxelles, Belgium}

%\thanks{}


\date{\today}

\begin{abstract}
Bell inequalities have traditionally been used to demonstrate that quantum theory is nonlocal, in the sense that there exist correlations generated from composite quantum states that cannot be explained by means of local hidden variables. With the advent of device-independent quantum information protocols, Bell inequalities have gained an additional role as certificates of relevant quantum properties. In this work we consider the problem of designing Bell inequalities that are tailored to detect maximally entangled states. We introduce a class of Bell inequalities valid for an arbitrary number of measurements and results, derive analytically their tight classical, non-signalling and quantum bounds and prove that the latter is attained by maximally entangled states. Our inequalities can therefore find an application in device-independent protocols requiring maximally entangled states.
\end{abstract}


\maketitle

\section{Introduction}
Measurements on separated subsystems in a joint entangled state may display correlations that cannot be mimicked by local hidden variable (LHV) models. 
These correlations are termed nonlocal and are detected by violating Bell inequalities \cite{bell1964,review}. In recent years it has become clear that non-locality is interesting not only for fundamental reasons, but also as a resource for device-independent (DI) quantum information tasks \cite{review} such as quantum key distribution
\cite{mayersyao,acin07} or random number generation 
\cite{colbeck,%*colbeckkent,
pironionature}. Thus, violations of Bell inequalities are not only indicators of non-locality, but can also be used to make qualitative and quantitative statements about operationally relevant quantum properties.

Traditionally, the problem of constructing Bell inequalities has been addressed from the point of view of deriving constraints satisfied by LHV models. Following this standard approach, the inequalities are derived using well-known techniques in convex geometry. Indeed, the set of correlations admitting LHV models defines a polytope \cite{review}, i.e., a bounded convex set with a finite number of vertices. These vertices correspond to local deterministic assignments, while the facets are the desired Bell inequalities. Facet (or tight) Bell inequalities provide necessary and sufficient criteria to detect the non-locality of given correlations. Clauser-Horne-Shimony-Holt (CHSH)~\cite{chsh1969} and Collins-Gisin-Linden-Massar-Popescu (CGLMP)~\cite{cglmp} Bell inequalities are examples thereof. 

Although such facet Bell inequalities are optimal detectors of non-locality, they are not necessarily optimal for inferring specific quantum properties in the DI setting. For instance, in a scenario where two binary measurements are performed on two entangled subsystems, it is well known that the violation of the CHSH inequality \cite{chsh1969} is a necessary and sufficient condition for non-locality. But certain ``non-facet" Bell inequalities are better certificates of randomness than the CHSH one when the two quantum systems are partially entangled \cite{acin2012}.

The main aim of this work is to introduce Bell inequalities valid for an 
arbitrary number of measurements and outcomes
whose maximal quantum violation, usually referred to as the \textit{Tsirelson bound}
\cite{tsirelson1980}, is attained by maximally entangled states 
\begin{equation}
|\psi^+_d \rangle=(1/\sqrt{d})\sum_{i=0}^{d-1}|ii\rangle.
\end{equation}
This is a desirable property since these states have particular features such as perfect correlations between outcomes of local measurements in the same bases, and therefore many quantum information protocols rely on them. In the particular case of two measurements CHSH is the simplest example of a Bell inequality with the above property, but others are known \cite{wonminson,lee2007,devicente2015} (see also results for many settings \cite{Ji,Liang,Lim}). Our construction works, however, for arbitrary numbers of measurements and outcomes, and, crucially, the Tsirelson bound of the resulting Bell inequalities can be computed \textit{analytically}.

In the case where only two measurements are made on each subsystem, all facet Bell inequalities are known for a small number of outputs and they are of the CGLMP form \cite{cglmp}. 
However, they are not maximally violated by the maximally entangled states of two qudits  (except in the case $d=2$ corresponding to the CHSH inequality) \cite{acindurt,zohrengill,npa2007}. Thus, we should not expect our Bell inequalities to be tight, and indeed they are not.

This implies that we cannot use standard tools from convex geometry to construct them. In fact, no quantum property is used for the construction of tight Bell inequalities like the CGLMP one and, in this sense, it is not surprising that their maximal violation does not require maximal entanglement. Our approach is completely different: it starts from quantum theory and exploits the symmetries and perfect correlations of maximally entangled states to derive a Bell inequality (cf. Ref. \cite{Lim} for a similar method). It exploits sum of squares decompositions of Bell operators, which is used to determine their Tsirelson bound. Thus, 
contrary to any previous derivation of Bell inequalities, 
quantum theory becomes a key ingredient of our method.

Our results provide new insight into the structure of the boundary of the set of quantum correlations (see discussion in Appendix \ref{appendixstructure}). In addition, our Bell inequalities have the potential to be used in DI quantum information protocols
such as random number generation, quantum key distribution or to self-testing \cite{mckague2012} (see more detailed discussion in Section \ref{secappli}).

\section{Preliminaries} We consider a Bell scenario with two distant parties $A$ and $B$ performing one of $m$ measurements $A_x$ and $B_y$ 
with $d$ outcomes on their share of some physical system. 
We label the measurements and outcomes as $x,y \in \{1,\ldots, m\}$ and $a,b \in \{0,\ldots, d-1\}$. The correlations obtained in this experiment are described by $(md)^2$ joint probabilities $P(A_x=a,B_y=b)$ that $A$ and $B$ obtain $a$ and $b$ upon performing the $x$th and $y$th measurement, respectively. These probabilities are ordered into a vector 
%
%\begin{equation}
$\vec{p}:=\{P(A_x=a,B_y=b)\}_{a,b,x,y} \in \mathbbm{R}^{(md)^2}.$
%\end{equation}

Importantly, the set of allowed vectors $\vec{p}$ varies depending on the physical principle they obey. 
If the measurements define spacelike separated events, the observed correlations
should obey the \textit{no-signalling principle}, which prevents
any faster-than-light communication among the parties.
These correlations form a convex polytope denoted $\mathcal{N}$. Contained in this set is 
the set of quantum correlations $\mathcal{Q}$ which is formed by 
those $\vec{p}$ whose components can be written as 
$P(A_x=a,B_y=b)=\langle \psi|P_a^{(x)}\otimes P_b^{(y)}|\psi\rangle$, 
where $|\psi\rangle$ is some state in a product Hilbert space $H_A\otimes H_B$ of unconstrained dimension, and $\{P_a^{(x)}\}$ and $\{P_b^{(y)}\}$ are projection operators defining, respectively, Alice's and Bob's measurements.
Finally, the set of correlations admitting LHV models, denoted $\mathcal{L}$, contains those $\vec{p}$ that can be written as a convex sum of product deterministic correlations $P(A_x=a,B_y=b)=P(A_x=a)P(B_y=b)$
with $P(A_x=a),P(B_y=b)=0,1$ for all $x,y$ \cite{fine1982}.

Bell was the first to prove that not all quantum correlations admit an LHV model \cite{bell1964}. To this end, he used the concept of a Bell inequality $I\leq C_b$, where 
$I$ is the so-called Bell expression that is a linear combination of the $(md)^2$ joint probabilities of the form
\begin{equation}\label{BE}
I:=\sum_{abxy} I_{abxy} P(A_x=a,B_y=b),
\end{equation}
and $C_b=\max_{\vec{p} \in \mathcal{L}}I$ is its classical bound.
The quantum or Tsirelson bound of $I$ is the maximum value $Q_b=\max_{\vec{p} \in \mathcal{Q}}I$ that it
can achieve for quantum correlations.
A Bell expression $I$ gives rise to a proper Bell inequality---one that is violated by quantum theory---if $C_b < Q_b$. 
If $\vec{p}$ violates a Bell inequality, the correlations described by $\vec{p}$ are termed nonlocal.
Finally, one defines $NS_b=\max_{\vec{p} \in \mathcal{N}}I$ as the maximum value of
$I$ over no-signalling correlations. For most of the known Bell inequalities 
$NS_b>Q_b>C_b$ \cite{bell1964,PR,ravi}.

Let us stress that although $\mathcal{Q}$ is convex, it is not a polytope. More importantly, the boundary of $\mathcal{Q}$ remains unknown 
despite several attempts to characterize it analytically
\cite{pawlowski09,wunderlich,lorthogonality,navascues2015} (see, nevertheless, \cite{Masanes}). 
This clearly makes the derivation of Tsirelson bounds a hard task. Given a Bell inequality, there is no procedure that guarantees finding its quantum bound, and it was achieved analytically only in a handful of cases. There is, however, a practical approximation scheme based on semidefinite programming, 
which consists in a hierarchy of sets $\mathcal{Q}_{1} \supseteq \mathcal{Q}_{2} \supseteq \dots \supseteq \mathcal{Q}_{k} \supseteq \dots$ converging to $\mathcal{Q}$ as $k \to  \infty$, and allows one to bound $Q_b$ from above \cite{npa2007} (see also \cite{devicente2015}).
Although for small Bell scenarios this method yields good numerical bounds (often tight), it becomes computationally expensive for scenarios involving a large number of measurements or outcomes.

\section{Class of Bell expressions}\label{secclass}
Our aim now is to introduce a family of Bell expressions, whose maximal quantum value is attained by the \textit{two-qudit} maximally entangled state $|\psi^+_d \rangle$. To derive them, we start from the premise that their maximal quantum values are obtained when Alice and Bob perform the optimal CGLMP measurements introduced in \cite{kaszlikowski,cglmp,barrett2006} (cf. Appendix \ref{appendixmeasurements}). This choice stems from the fact that these measurements generalize the CHSH measurements ($d=2$) to arbitrary dimensions and they lead to non-local correlations that are most robust to noise~\cite{kaszlikowski} or for $m=2$ give a stronger statistical test~\cite{gill}. 

The probabilities $P(A_x=a,B_y=b)$ obtained when using the optimal CGLMP measurements on $|\psi^+_d \rangle$ have several symmetries. For instance, they only depend on the difference $a-b=k\mod d$. If we impose that our Bell expressions respect this symmetry, the probabilities $P(A_{x} = j + k \text{ mod }d,\text{ } B_{y} = j)$ should be treated equally for all $j$, i.e., the Bell expressions should be linear combinations of $P(A_{x} = B_{y} + k) := \sum_{j = 0}^{d-1} P(A_{x} = j + k \text{ mod }d,\text{ } B_{y} = j)$. Taking into account all symmetries, a generic form for our Bell expressions is 
\begin{equation}\label{Bellproba}
I_{d,m}:=\sum_{k=0}^{\left \lfloor d/2\right \rfloor-1}\left(
\alpha_k \mathbbm{P}_k -\beta_k \mathbbm{Q}_k \right ),
\end{equation}
where $\mathbbm{P}_k := \sum_{i=1}^{m}[P(A_i=B_i+k)+P(B_i=A_{i+1}+k)]$, $\mathbbm{Q}_k := \sum_{i=1}^{m}[P(A_i=B_i-k-1)+P(B_i=A_{i+1}-k-1)]$ with $A_{m+1}:=A_1+1$. 
The parameters $\alpha_k$ and $\beta_k$ are our degrees of freedom. Taking, e.g., $\alpha_k = \beta_k = 1 - 2k/(d - 1)$ for $m=2$, one recovers the CGLMP Bell inequalities.

To exploit the symmetries inherent in Bell inequalities, we often write them in terms of correlators instead of probabilities. As we consider an arbitrary number of outcomes,
we appeal to the notion of generalized correlators (see, e.g., Ref. \cite{Liang} and Ref. \cite{bancal12} for other options).
%While in the CHSH case these correlators are very simply defined, here if we wish to preserve the amount of information we need to use the Fourier representation of Bell inequalities.
These are complex numbers that are defined through the two-dimensional 
Fourier transform of the probabilities $P(A_x=a,B_y=b)$:
\begin{equation}
 \langle A_x^kB_y^l\rangle=\sum_{a,b=0}^{d-1}\omega^{ak+bl}P(A_x=a,B_y=b),
\label{correlators}
\end{equation}
where $\omega = \text{exp}(2\pi i /d)$, $k,l\in \{0,\ldots,d-1\}$, and 
$\{A_x^k\}_k$ and $\{B_y^l\}_l$ can be thought of as measurements with outcomes labelled by roots of unity $\omega^i$ $(i=0,\ldots,d-1)$.
For quantum correlations $\vec{p}$, the correlators $\langle A_x^kB_y^l\rangle$ are average values of the tensor product of the operators
%
%\begin{equation}\label{operators}
$A_x^{k}=\sum_{a=0}^{d-1}\omega^{ak}P_a^{(x)}\quad \mathrm{and} \quad B_y^{l}=\sum_{b=0}^{d-1}\omega^{bl}P_b^{(y)}$
%\end{equation}
%
in the state $\ket{\psi}$. 
Note that they are unitary, their eigenvalues are the roots of unity, and they satisfy $(A_x^{k})^{\dagger}=A_x^{d-k}$ and $(B_y^{l})^{\dagger} = B_{y}^{d-l}$ for any $k,l$.

Now, exploiting (\ref{correlators}), expression (\ref{Bellproba}) 
can be rewritten as
\begin{equation}\label{Bellcorr}
 \widetilde{I}_{d,m}= \sum_{i=1}^m\sum_{l=1}^{d-1}\langle A_i^l\bar{B}_i^l \rangle,
\end{equation}
where, for clarity, the change of variables $\bar{B}_i^l = a_l B_i^{d-l}+a_l^{*}B_{i-1}^{d-l}$ with $ a_l=\sum_{k=0}^{\lfloor d/2\rfloor-1}(\alpha_k\omega^{-kl}-\beta_k\omega^{(k+1)l})$ was introduced on Bob's side. Due to the convention $A_{m+1}=A_1+1$, the term $\bar{B}_1^l$ is defined as $\bar{B}_1^l=a_lB_{1}^{d-l}+a_l^* \omega^lB_{m}^{d-l}$. 
For simplicity, in (\ref{Bellcorr}) we ignored the irrelevant scalar term corresponding to $l=0$ and rescaled the expression.
Below we denote the classical, quantum and no-signaling bound of $\widetilde{I}_{d,m}$ by
$\widetilde{C}_b$, $\widetilde{Q}_b$ and $\widetilde{NS}_b$, respectively.

Our aim now is to fix the free parameters $\alpha_k$ and $\beta_k$ according to the quantum property we need: maximal violation by the maximally entangled state $| \psi^+_d \rangle$. At this point, it is instructive to look at the specific example of the CHSH Bell expression ($m=2$, $d=2$). In the notation (\ref{Bellcorr}) 
the CHSH Bell expression $\langle A_1B_1\rangle + \langle A_1 B_2\rangle + \langle A_2 B_1\rangle- \langle A_2 B_2\rangle$ reads
%
%\begin{equation}
$\widetilde{I}_{2,2}=\langle A_1 \bar{B}_1\rangle +\langle A_2 \bar{B}_2\rangle,$
%\end{equation}
where $\bar{B}_1=(B_1 + B_2)/\sqrt{2}$, $\bar{B}_2=(B_1 - B_2)/\sqrt{2}$. Then, for the optimal measurements leading to the Tsirelson bound of $\widetilde{I}_{2,2}$, we have $\bar{B}_1=A_1^*$ and $\bar{B}_2=A_2^*$. This reflects the property that
for the maximally entangled state
\begin{equation}\label{Raimat}
M\otimes N\ket{\psi_d^+}=\mathbbm{1}\otimes NM^{T}\ket{\psi_d^+},\quad \forall M,N.
\end{equation}
This condition implies that a measurement by Alice is perfectly correlated 
with its complex conjugate by Bob.
Our intuition to derive Bell inequalities detecting maximal entanglement is to impose this property for any $m$ and $d$: we choose the parameters $\alpha_k$ and $\beta_k$ such that 
\begin{equation}\label{conditions}
\bar{B}_{i}^{l} = (A_{i}^{l})^{*}
\end{equation}
hold for $l = 1,\ldots, d -1$ and $i = 1, \ldots, m$ with the initial operators $\{P_a^{(x)}\}$ and $\{P_b^{(y)}\}$ being the optimal CGLMP operators. 
Conditions (\ref{conditions}) give rise to a set of linear equations for 
$\alpha_k$ and $\beta_k$ which yields (see Appendix \ref{appendixcoefficients} for details)
\begin{equation}\label{alpha}
\alpha_k= \frac{1}{2d}\tan\left(\frac{\pi}{2m}\right) \left[g(k) - g\left(\left\lfloor \frac{d}{2} \right\rfloor\right)\right],
\end{equation}
\begin{equation}\label{beta}
\beta_k= \frac{1}{2d}\tan\left(\frac{\pi}{2m}\right) \left[g\left(k + 1 - \frac{1}{m}\right) + g\left(\left\lfloor \frac{d}{2} \right\rfloor\right)\right]
\end{equation}
with $g(x):=\cot(\pi(x+ 1/2m)/d)$.

To sum up, our class of Bell expressions is given by $I_{d,m}$ (\ref{Bellproba}) or equivalently by $\widetilde{I}_{d,m}$ (\ref{Bellcorr}), with coefficients (\ref{alpha}) and (\ref{beta}). We arrived at it by writing the most general Bell expression satisfying the symmetry of CGLMP correlations, re-writing these Bell expressions in the simple form (\ref{Bellcorr}) through a change of variable on Bob's side, and then imposing the conditions (\ref{conditions}) that take into account the symmetries of the maximally entangled state, as CHSH does for two binary measurements.

\section{Properties of the novel Bell expressions}\label{secproperties}
We now  
analyze the main properties of our Bell expressions: we compute all the relevant bounds $\widetilde{C}_b$, $\widetilde{Q}_b$, $\widetilde{NS}_b$, and show that 
$\widetilde{C}_b<\widetilde{Q}_b<\widetilde{NS}_b$ for any $d$ and $m$. For clarity we only include sketches of proofs (see Appendices \ref{appendixclassical}, \ref{appendixtsirelson} and \ref{appendixns} for details).

Let us begin with the classical bound.

\begin{thm}\label{theoclass}
The classical bound of $\widetilde{I}_{d,m}$ is given by $\widetilde{C}_{b} = (1/2) \tan \left(\pi/2m\right) \left\{ (2m - 1) g(0) - g(1 - 1/m)\right\} - m.$
\end{thm}
\begin{proof}
We start with the expression $I_{d,m}$. Since we can restrict the problem to local deterministic strategies, finding $\widetilde{C}_b$ becomes a question of distributing $0$s and $1$s over all the terms $P(A_x = B_y + z)$. It turns out that the optimal strategy is to set $2m - 1$ of the terms multiplied by $\alpha_0$ and a single term multiplied by $\beta_0$ to one, and the remaining terms to zero.
\end{proof}

Importantly, the resulting Bell inequality $\widetilde{I}_{d,m}\leq \widetilde{C}_b$ is violated by quantum theory; one can reach the value $\widetilde{I}_{d,m}=m(d-1)$ by applying the CGLMP measurements on $\ket{\psi^+_d}$. This is seen by using Eq. (\ref{conditions}), the unitarity of $A_{i}^{k}$, and the symmetries of the maximally entangled states (\ref{Raimat}). Then, all the correlators in (\ref{Bellcorr}) equal one, yielding the quantum violation of $m(d-1)$.
This violation is optimal and defines the tight Tsirelson bound of $\widetilde{I}_{d,m}$.

\begin{thm}\label{theoquantum}
The Tsirelson bound of $\widetilde{I}_{d,m}$ is given by $\widetilde{Q}_{b} = m(d-1)$. 
\end{thm}
\begin{proof}
We construct a sum-of-squares (SOS) decomposition of the shifted Bell operator $\widetilde{\mathcal{B}} := \widetilde{Q}_{b}\mathbbm{1} - \mathcal{B}$, where $\mathbbm{1}$ is the identity operator and $\mathcal{B}$ the Bell operator corresponding to expression (\ref{Bellcorr}) (see, e.g., \cite{bamps2015,doherty}). For any positive semidefinite operator $\mathcal{P}$, an SOS decomposition is a collection of operators ${P_{\lambda}}$ such that
%%
%\begin{equation}
$\mathcal{P} = \sum_{\lambda} P_{\lambda}^{\dagger}  P_{\lambda}.$
%\label{sostheo}
%\end{equation}
%%
If $\widetilde{\mathcal{B}}$ admits the latter form it must be positive semidefinite, implying that $\widetilde{Q}_{b}$ upper bounds our Bell expression, i.e., $\langle \psi | \mathcal{B} | \psi \rangle \leq \widetilde{Q}_{b}$ for any $\ket{\psi}$. 
%
%This approach is in principle valid for any shifted Bell operator, thus for any Bell %expression. As we expect the $P_{\lambda}$'s to be polynomials of the measurement %operators of Alice and Bob, we can define the order of the SOS decomposition as the %largest degree of these polynomials. 

To show that $\widetilde{Q}_{b} = m(d-1)$ is indeed the Tsirelson bound of $\widetilde{I}_{d,m}$, 
we prove that $\widetilde{Q}_{b}\mathbbm{1}-\mathcal{B}$ decomposes as
\begin{equation}\label{sos}
\widetilde{Q}_{b}\mathbbm{1}-\mathcal{B}  = \frac{1}{2}\sum_{i=1}^{m}\sum_{k=1}^{d-1} P_{ik}^{\dagger} P_{ik} +  \frac{1}{2}\sum_{i=1}^{m - 2}\sum_{k=1}^{d-1} T_{ik}^{\dagger} T_{ik},
\end{equation}
where %$P_{ik} = (\mathbbm{1}-A_i^k\ot \bar{B}_i^k)$, 
$P_{ik} = \mathbbm{1} \otimes \bar{B}_i^k - (A_i^{k})^{\dagger} \otimes \mathbbm{1}$, and $T_{ik} = ( \mu_{i,k} B_{2}^{d - k} + \nu_{i,k} B_{i + 2}^{d - k} + \tau_{i,k} B_{i + 3}^{d - k})$ with $\mu_{i,k}$, $\nu_{i,k}$, $\tau_{i,k} \in \mathbb{R}$. The Bell operator reads $\mathcal{B} = \sum_{i=1}^m\sum_{l=1}^{d-1} A_i^k\otimes\bar{B}_i^k$, and the decomposition is independent of the choice of $A_i^k$ and $B_i^k$. 
The exact values of the coefficients along with details on the SOS decomposition can be found in Appendix \ref{appendixtsirelson}. 
\end{proof}
A few remarks are in order. First, it is not difficult to see that $\widetilde{Q}_b>\widetilde{C}_b$ for any $m,d\geq 2$, meaning that 
all our Bell inequalities are nontrivial (cf. Appendix \ref{appendixscalings}). Second, let us elaborate on how the SOS works in the case of two measurements, $m=2$, which justifies
the choice of conditions (\ref{conditions}). For $m=2$, the second part of the SOS decomposition \eqref{sos} vanishes. %It uses the following property of the maximally entangled states : $M \otimes N \ket{\psi^{+}} = \mathbb{I} \otimes N M^{T} \ket{\psi^{+}}$ for $M$ and $N$ operators. 
For the optimal CGLMP measurements both sides of (\ref{sos}) must yield zero when applied to $\ket{\psi^+_d}$, which stems from conditions (\ref{Raimat}) and (\ref{conditions}). 
This allows one to grasp the intuition behind conditions (\ref{conditions}), i.e., they allow one to construct in a quite direct way
an SOS decomposition (\ref{sos}), in which all operators $P_{ik}$ are polynomials of the measurement operators $A_{i}^{k}$ and $B_{i}^{k}$ of order one, significantly facilitating the computation of the Tsirelson bound.
For the CHSH Bell inequality, one observes the same effect, as these same properties of the optimal state and measurements allow 
the Bell operator $\mathcal{B}_{\text{CHSH}}  = A_1 \otimes B_1 + A_1 \otimes B_2 + A_2 \otimes B_1 - A_2 \otimes B_2$ to have the decomposition:
%
%\begin{equation}\label{soschsh}
$2\sqrt{2}\mathbbm{1} - \mathcal{B}_{\text{CHSH}} = ( P_{1}^{\dagger}P_1 + P_{2}^{\dagger}P_2)/\sqrt{2},$
%\end{equation}
%
with $P_1 = (1/\sqrt{2})\mathbbm{1}\otimes (B_1 +B_2) - A_1 \otimes \mathbbm{1}$, and $P_2 = (1/\sqrt{2})\mathbbm{1} \otimes(B_1 - B_2) - A_2 \otimes \mathbbm{1} $. Thus, our construction generalizes this quantum aspect of the CHSH Bell operator. For larger number of measurements, $m>2$, the first part of the SOS decomposition is not enough and one has to add 
``by hand'' the extra term in which all $ T_{ik}$'s are also of order one in $B_i^k$.

Note that for two measurements, our Bell expressions coincide with those introduced in \cite{wonminson} and then rederived in \cite{devicente2015} using a different approach. Moreover, the Tsirelson bounds of these Bell inequalities was computed in Refs.
\cite{lee2007,devicente2015} exploiting other techniques, and it was proven in \cite{lee2007} that they
are not tight. On the other hand, for $d=2$ and any $m$, our class recovers the 
well-known chained Bell inequalities \cite{chained}. We finally notice that the alternative 
generalization of the CHSH Bell inequality to three measurements and outcomes given in \cite{BuhrmanMassar} 
was also found to be maximally violated by $\ket{\psi^+_3}$ \cite{Liang}.

Let us eventually compute the no-signalling bound of our Bell expressions.

\begin{thm}\label{theons}
The no-signalling bound of $\widetilde{I}_{d,m}$ is given by $\widetilde{NS}_{b} = m \tan \left(\pi/2m\right) g(0) - m$.
\end{thm}
\begin{proof}
We provide no-signalling correlations $\vec{p}$ and show that they attain the algebraic bound of $I_{d,m}$. 
They correspond to having all the probabilities which are multiplied by $\alpha_0$ in $I_{d,m}$ equal to one, and all the others equal to zero (see Appendix \ref{appendixns}).
\end{proof}

Again, it is not difficult to see that $\widetilde{NS}_b>\widetilde{Q}_b$ for any $m,d\geq 2$
(see Appendix \ref{appendixscalings} for the proof and scalings of $\widetilde{C}_b$, $\widetilde{Q}_b$
and $\widetilde{NS}_b$ with $m$ and $d$).

\section{Applications to device-independent protocols}\label{secappli}
A natural application for our Bell inequalities is self-testing---a DI protocol in which a state and measurements performed on it are certified up to local isometries, based on the nonlocal correlations they produce. To perform self-testing, the correlations $\vec{p}$ maximally violating the given Bell inequality must be unique, i.e., attained, up to local isometries, by certain state and measurements. This is generally hard to prove. There exists, however, a numerical method for self-testing \cite{swaptest1}. We applied it to the simplest case $m=2$ and $d = 3$, and the results are plotted in Figure \ref{swap3}.
It shows that one can self-test the maximally entangled state of two qutrits $\ket{\psi^+_3} =(\ket{00} + \ket{11} + \ket{22})/\sqrt{3}$ with our inequalities.

\begin{figure}
\includegraphics[width=0.4\textwidth]{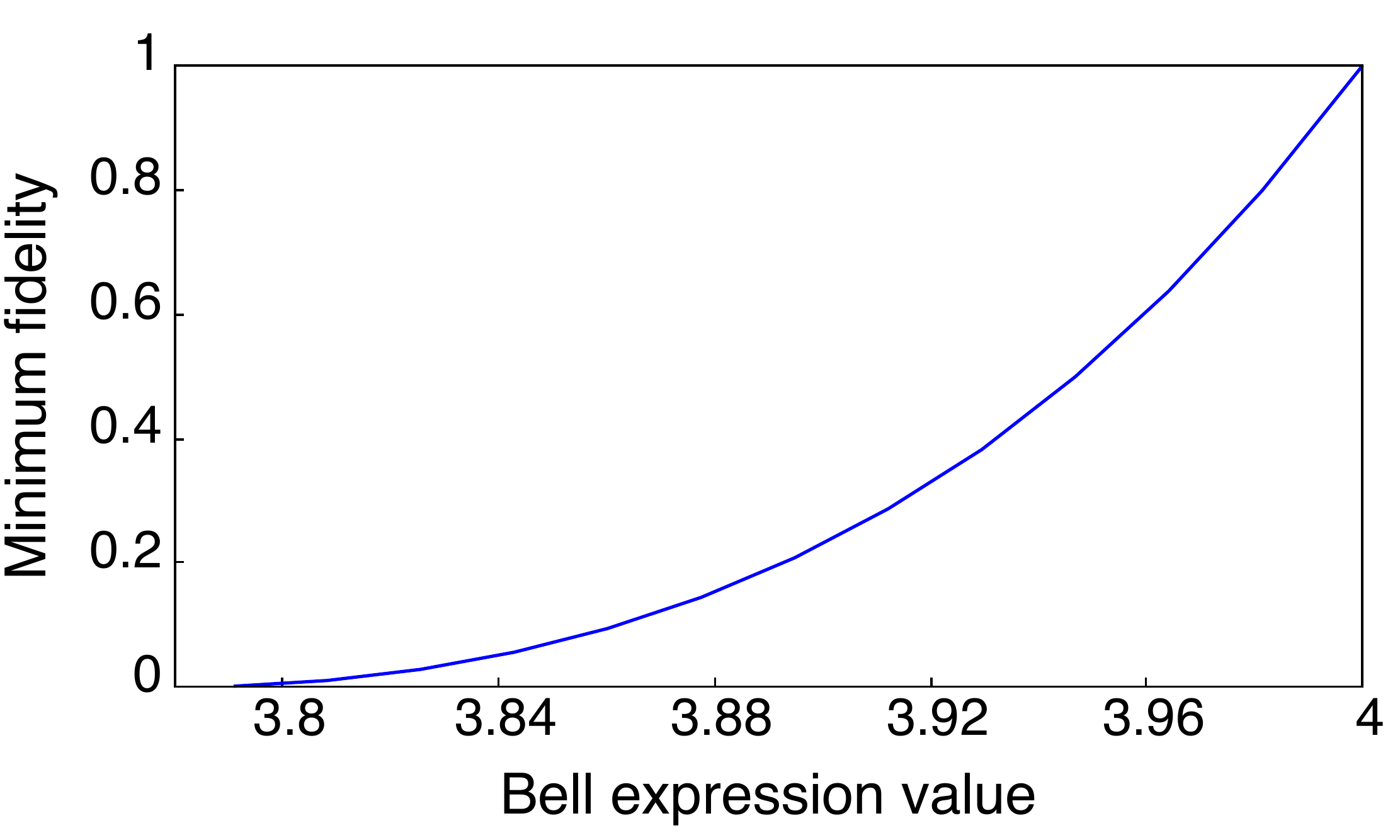}
\caption{Minimum fidelity of the state in the black box to the maximally entangled state of two qutrits, as a function of the violation of $\widetilde{I}_{3,2}$. At the maximal violation $4$, the fidelity is equal to $1$, meaning that the quantum state used in the Bell experiment must be maximally entangled. The numerical method that we used does not yield a positive lower bound on the fidelity below $ \widetilde{I}_{3,2}  \approx 3.79$ (for comparison, the classical bound is $\widetilde{I}_{3,2}  \leq (1 + 3\sqrt{3})/2 \approx 3.01$).\label{swap3}}
\end{figure}

An open question is whether one can generalize this result to any dimension. Our inequalities could then be applied in DI random number generation protocols \cite{colbeck,%*colbeckkent,
pironionature,pironiomassar}. Indeed, if $\vec{p}$ maximally violating $\widetilde{I}_{d,m}$ is unique, one can apply the method of \cite{dhara2013} and use the symmetries of the Bell expressions to guarantee a dit of perfect randomness. This, by increasing the dimension $d$, would result in unbounded randomness expansion.

Our inequalities could also find applications in DI quantum key distribution. An advantage that our inequalities have over CGLMP in that scenario \cite{marcusmarcin} is that, as said before, the maximal violation is obtained for the maximally entangled state. This state can produce perfect correlations between the users, which reduces the error-correcting phase of the protocol and can lead to better key generation rates. We study this question in Appendix \ref{diqkd}. Numerically, we find that for $m=2$ and $d=3$, our inequalities lead to higher key rates than CGLMP for levels of white noise up to $4.2$ percent. While this advantage is not very large, we believe it grows with the dimension of the systems, at least in the noiseless case \cite{zohrengill}. %, because the entanglement entropy of the state maximally violating CGLMP decreases as a function of $d$. 
Moreover, it is known that maximally entangled states are much simpler to prepare experimentally than fine-tuned partially entangled states. It would be interesting to confirm these conjectures in a future work focused on DIQKD.

\section{Conclusions}
In this work, we introduced a new technique allowing to construct Bell inequalities with arbitrary numbers of measurements and outcomes that are maximally violated by the maximally entangled states. It
exploits the SOS decompositions of Bell operators and, crucially, 
allows one to compute analytically their Tsirelson bounds.  
Our results are general as, unlike previous works, we do not consider a particular Bell scenario, but allow for arbitrary number of measurements $m$ and outcomes $d$. Our inequalities can be seen as the ``quantum'' or the DI-oriented generalization of CHSH Bell inequality, in the same spirit as the CGLMP inequality generalizes the CHSH one classically. 
%\newline

\begin{acknowledgments}
%\textit{Acknowledgments.} 
We wish to thank Y.-C. Liang, M. Navascu\'{e}s, T. V\'{e}rtesi and J. Kaniewski for fruitful discussions, and especially J.-D. Bancal for sharing with us his code. This work was supported ERC CoG QITBOX and AdG OSYRIS, the AXA Chair in Quantum Information Science, Spanish MINECO (FOQUS FIS2013-46768-P, SEV-2015-0522,  QIBEQI FIS2016-80773-P and  FISICATEAMO FIS2016-79508-P), Fundaci\'o Privada Cellex, the Generalitat de Catalunya (SGR 874, SGR 875  and the CERCA programme), the EU projects QALGO and SIQS, and the John Templeton Foundation. This project has received funding from the European Union’s Horizon 2020
research and innovation programme under the Marie Sklodowska-Curie grant
agreements No 705109 and No 748549. We acknowledge financial support from the Fondation Wiener-Anspach and the Interuniversity Attraction Poles program of the Belgian Science
Policy Office under the grant IAP P7-35 photonics@be. J. T. acknowledges the CELLEX-ICFO-MPQ programme. S. P. is a Research Associate of the Fonds de la Recherche Scientifique F.R.S.-FNRS (Belgium). 
\end{acknowledgments}

\onecolumngrid
\appendix

\section{Optimal CGLMP measurements}\label{appendixmeasurements}
We present here the ``optimal CGLMP measurements'' first introduced in \cite{kaszlikowski} and generalized to an arbitrary number of inputs in \cite{barrett2006}, as we use them throughout our work. They are defined as follows 
\begin{equation}
A_{x} = U_{x}^{\dagger}F \Omega F^{\dagger} U_{x}, \qquad \qquad B_{y} = V_{y}F^{\dagger} \Omega F V_{y}^{\dagger},
\label{cglmpmeas}
\end{equation}
where $\Omega=\mathrm{diag}[1,\omega,\omega^2,\ldots,\omega^{d-1}]$, with $\omega = \text{exp}(2\pi i/d)$, and $F$ is the $d\times d$ discrete Fourier transform matrix given by
\begin{equation}\label{AndresSegovia}
F_d=\frac{1}{\sqrt{d}}\sum_{i,j=0}^{d-1}\omega^{ij}\ket{i}\!\bra{j}.
\end{equation}
Then, $U_x$ and $V_x$ are unitary operations defining Alice's and Bob's measurements
and read explicitly
\begin{equation}\label{Alhambra}
U_x=\sum_{j=0}^{d-1}\omega^{j \theta_x}\proj{j},\qquad V_y=\sum_{j=0}^{d-1}\omega^{j \zeta_y}\proj{j}
\end{equation}
with the phases $\theta_{x} = (x - 1/2)/m$ and $\zeta_{y} = y/m$ for $x, y = 1, \ldots, m$.

When applying these measurements on a normalised state of the form $|\psi \rangle = \sum_{q = 0}^{d-1} \gamma_q |qq\rangle$, we obtain the probabilities
\begin{equation}\label{probasmeas}
P(A_x = a, B_y = b) = \left| \frac{1}{d} \sum_{q = 0}^{d-1} \gamma_q \exp \left( \frac{2\pi i}{d} q (a - b - \theta_{x} + \zeta_{y}) \right) \right|^{2}.
\end{equation}
One can observe that this depends only on the difference $k = a -b$ and not on $a$ and $b$ separately. This means that:
\begin{equation}
P(A_x = B_y + k) = d P(A_x = k, B_y = 0).
\end{equation}
Thus, all the terms $P(A_x = B_y + k)$ computed for those measurements and state have identical subterms $P(A_x = b + k, B_y = b)$. Moreover, using the values of the phases $\theta_{x}$ and $\zeta_{y}$, one can verify straightforwardly that expression (\ref{probasmeas}) has the same value if $x = y$ and $a -b = k$, and if $x = y + 1$ and $a -b =- k$. Thus :
\begin{equation}
P(A_i = B_i + k) = P(B_i  = A_{i+1} + k),
\end{equation}
for $i = 1, \ldots, m$. Note that if one wishes to write $A_{m+1} = A_{1}$, the symmetry is not valid anymore and requires the definition  $A_{m+1} = A_{1} + 1$, which we adopt. To sum up, all the $ \mathbbm{P}_k$ and $\mathbbm{Q}_k$ from our class of inequalities have identical subterms for those state and optimal CGLMP measurements (in particular the state can be the maximally entangled state). These symmetries justify the form of our Bell expressions: terms who have the same value appear with the same coefficient $\alpha_k$ or $\beta_k$, thus forming ``blocks''. Different blocks have different values and are multiplied by different coefficients.

\section{Derivation of coefficients $\alpha_k$ and $\beta_k$ }\label{appendixcoefficients}
We present the details on the derivation of coefficients $\alpha_k$ and $\beta_k$ whose value is stated in Section \ref{secclass} of the main text. The departure point of the determination of $\alpha_k$ and $\beta_k$ is the set of matrix conditions
\begin{equation}\label{conditions2}
\bar{B}_{i}^{l} = (A_{i}^{l})^{*}
%\tag{\ref{conditions}}
\end{equation}
with $i=1,\ldots,m$, and $l=1,\ldots,\lfloor d/2\rfloor$. This number $\lfloor d/2\rfloor$ of equations stems from the fact that 
\begin{eqnarray}
A_x^{d-l}=(A_x^l)^{\dagger}, \label{aproperties} \\
\bar{B}_y^{d-l}=(\bar{B}_y^l)^{\dagger}. \label{bproperties}
\end{eqnarray}
Recall that the barred quantities $\bar{B}_i^l$ are defined as
\begin{equation}\label{Bbar}
\bar{B}_i^l=a_lB_{i}^{d-l}+a_l^* B_{i-1}^{d-l}
\end{equation}
for $i=2,\ldots,m$ and $\bar{B}_1^l=a_lB_{1}^{d-l}+a_l^* \omega^lB_{m}^{d-l}$, 
and the numbers $a_l$ are given by 
\begin{equation}\label{system2}
a_l=\sum_{k=0}^{\lfloor d/2\rfloor-1}\left[\alpha_k\omega^{-kl}-\beta_k\omega^{(k+1)l}\right].
\end{equation}
Notice that $a_l = a_{d-l}^{*}$. Let us notice in passing that the properties (\ref{aproperties}) and (\ref{bproperties}) imply that the Bell expression we consider, i.e., 
\begin{equation}\label{delicje}
\widetilde{I}_{d,m}=\sum_{i=1}^m\sum_{l=1}^{d-1}\langle A_i^l \bar{B}_i^l\rangle
\end{equation}
is real. This is because the sum in (\ref{delicje}) can be split into 
two sums: for $l=1,\ldots, \lfloor d/2\rfloor$ and $l=\lfloor d/2\rfloor+1,\ldots,d-1$
for odd $d$, and for $l=1,\ldots,d/2-1$ and $l=d/2+1,\ldots,d-1$ (plus a single term corresponding to $l=d/2$ which is always real) for even $d$. Now, due to Eqs. (\ref{aproperties}) and (\ref{bproperties}) one realizes that all terms in the second sum are complex conjugations 
of those in the first sum.
%
%\begin{eqnarray}
%a_l A_{i}^{l} \otimes B_{i}^{d-l} + a_{d-l} A_{i}^{d-l} \otimes B_{i}^{l}, \\
%a_l^{*} A_{i}^{l} \otimes B_{i-1}^{d-l} + a_{d-l}^{*} A_{i}^{d-l} \otimes B_{i-1}^{l}.
%\end{eqnarray}
%These can be rewritten respectively as
%\begin{eqnarray}
%a_l A_{i}^{l} \otimes B_{i}^{d-l} + a_{l}^{*} (A_{i}^{l} \otimes B_{i}^{d-l})^{\dagger}, \\
%a_l^{*} A_{i}^{l} \otimes B_{i-1}^{d-l} + a_{l} (A_{i}^{l} \otimes B_{i-1}^{d-l})^{\dagger},
%\end{eqnarray}
%which yield the sum of a complex number and its conjugate, i.e. a real number. For even $d$ and %$l = d/2$, there is only one coefficient $a_{d/2} = a_{d/2}^{*}$, which is real.

In order to solve the system (\ref{conditions2}) one has to find explicit forms of $A_x^l$ and $B_y^l$. Introducing Eqs. (\ref{AndresSegovia}) and (\ref{Alhambra}) into Eq. (\ref{cglmpmeas}),
one obtains
\begin{equation}\label{Axl}
A_x^l=\omega^{-(d-l)\theta_x}\sum_{n=0}^{l-1}\ket{d-l+n}\!\bra{n}+\omega^{l\theta_x}\sum_{n=l}^{d-1}\ket{n-l}\!\bra{n}
\end{equation}
%
%\begin{equation}
%A_{x}^{l} = \omega^{(d- l) \theta_{x}} \sum_{n = 0}^{l - 1}\ket{n}\!\bra{d-l+n}+\omega^{- l \theta_{x}} \sum_{n = l}^{d - 1} \ket{n}\!\bra{n %- l}
%\end{equation}
%
and
\begin{equation}\label{Byl}
B_y^l=\omega^{-(d-l)\zeta_y}\sum_{n=0}^{l-1}\ket{n}\!\bra{d-l+n}
+\omega^{l\zeta_y}\sum_{n=l}^{d-1}\ket{n}\!\bra{n-l}.
\end{equation}
%
%\begin{equation}
%B_{y}^{l} = \omega^{-(d- l) \zeta_{y}} \sum_{n = 0}^{l-1} \ket{d-l+n}\! \bra{n }+\omega^{l \zeta_{y}} \sum_{n = l}^{d - 1} \ket{n - l}\!%
%\bra{n} .
%\end{equation}
%
%Starting from (\ref{cglmpmeas}), one can obtain the form of operators $A_{x}^{l} = \omega^{- l \theta_{x}} \sum_{n = l}^{d - 1} \ket{n} 
%\bra{n - l} + \omega^{(d- l) \theta_{x}} \sum_{n = d -  l}^{d - 1} \ket{n - d + l} \bra{n}$, and $B_{y}^{l} = \omega^{l \zeta_{y}} \sum_{n = %l}^{d - 1} \ket{n - k} \bra{n} + \omega^{(d- l) \zeta_{y}} \sum_{n = d -  l}^{d - 1} \ket{n} \bra{n - d + l}$. 
%
%\added{(dodac o tym, ze $A_{m+1}=A_1+1$)}
Then, one combines these formulas with equations (\ref{Bbar}) and 
(\ref{conditions2}), and compares the matrix elements, which yields
the following system of equations
\begin{eqnarray}
 a_l\omega^{-l\zeta_i}+a_l^{*}\omega^{-l\zeta_{i-1}}&\!=\!&\omega^{-l\theta_i}
\nonumber\\
 a_l\omega^{(d-l)\zeta_i}+a_l^{*}\omega^{(d-l)\zeta_{i-1}}&\!=\!&\omega^{
(d-l)\theta_i },
\label{system}
\end{eqnarray}
with $i=1,\ldots,m$ and $l=1,\ldots,\lfloor d/2\rfloor$, where it is assumed that $\zeta_0=0$.
Simple algebra implies finally that
%One straightforwardly solves this system, obtaining $a_l=\omega^{(2l-d)/4m}/[2\cos(\pi/2m)]$
%with $l=1,\ldots,\lfloor d/2 \rfloor$, which after 
%
\begin{equation}
a_l=\frac{\omega^{\frac{2l-d}{4m}}}{2\cos(\pi/2m)}\qquad (l=1,\ldots,\lfloor d/2\rfloor).
\end{equation}

Having determined $a_l$, one can turn to the system (\ref{system2}). It consists of $\lfloor d/2\rfloor$ equations containing $2\lfloor d/2\rfloor$ variables, meaning that it cannot be uniquely solved, and, in particular, the solutions will be generally complex.
To handle the latter problem we equip this system with $\lfloor d/2\rfloor$ additional equations
\begin{equation}\label{system3}
\sum_{k=0}^{\lfloor d/2\rfloor-1}\left[\alpha_k\omega^{kl}-\beta_k\omega^{-(k+1)l}\right]=a_l^*.
\end{equation}
for $l=1,\ldots,\lfloor d/2\rfloor$.
Now, both systems (\ref{system2}) and (\ref{system3}) can be condensed into the following single one 
\begin{equation}\label{system4}
\sum_{k=0}^{\lfloor d/2\rfloor-1}\left[\alpha_k\omega^{-kl}-\beta_k\omega^{(k+1)l}\right]=c_l,
\end{equation}
in which $c_l=a_l$ for $l=1,\ldots,\lfloor d/2\rfloor$
and $c_{l}=c_{-l}^{*}$ for $l=-\lfloor d/2\rfloor,\ldots,-1$. In what follows we solve
(\ref{system3}) for even and odd $d$ separately.

\paragraph{Odd $d$.} We begin by noting that in this case, 
the system (\ref{system4}) consists of $d-1$ equations and involves the same number of
variables, and therefore one expects it to have a unique solution. To find it, we
denote the set $I:=\{-(d-1)/2,\ldots,-1,1,\ldots,(d-1)/2\}$ and note that 
for any pair $k,n\in \{0,\ldots,\lfloor d/2\rfloor-1\}$, the following identity holds:
\begin{equation}\label{GorzkaZoladkowa}
\sum_{l\in I}\omega^{-lk}\omega^{ln} =\sum_{l\in I\cup\{0\}}\omega^{-lk}\omega^{ln} - 1= d\delta_{n,k}-1.
\end{equation}
We then multiply (\ref{system4}) by $\omega^{nl}$ for some $n\in\{0,\ldots,\lfloor d/2\rfloor-1\}$ and add the resulting equations over $l\in I$, which by virtue of Eq. (\ref{GorzkaZoladkowa}) gives
\begin{equation}\label{alphan}
\alpha_n=\frac{1}{d}S+\frac{1}{d}\sum_{l\in I}c_l\omega^{nl}\qquad (n=0,\ldots,\lfloor d/2\rfloor-1),
\end{equation}
where we have denoted
\begin{equation}\label{S}
S=\sum_{k=0}^{\lfloor d/2\rfloor-1}(\alpha_k-\beta_k).
\end{equation}
The coefficients $\beta_n$ can be determined in an analogous way and we obtain:
%To this end, we multiply the system (\ref{system3}) by $\omega^{-(n+1)l}$ for some $n$, and then we add the resulting equations. Exploiting Eq. (\ref{GorzkaZoladkowa}) one arrives at
%
\begin{equation}\label{betan}
\beta_n=-\frac{1}{d}S-\frac{1}{d}\sum_{l\in I}c_l\omega^{-(n+1)l} \qquad (n=0,\ldots,\lfloor d/2\rfloor-1).
\end{equation}

To fully determine $\alpha_n$ and $\beta_n$,
%one still needs to compute the sums appearing in equations (\ref{alphan}) and (\ref{betan}) and the quantity $S$. 
it is in fact enough to compute the sum in Eq. (\ref{alphan}) as the second one and $S$ can be obtained from it by replacing
$n$ by $-(n+1)$ and $\lfloor d/2\rfloor$, respectively. To compute this sum, we first express it as
\begin{eqnarray}\label{Sonda}
\sum_{l\in I}c_l\omega^{nl}&=&\frac{1}{\cos(\pi/2m)}\sum_{l=1}^{\lfloor d/2\rfloor}\mathrm{Re}
\left(\omega^{(2l-d)/4m}\omega^{nl}\right)\nonumber\\
&=&\frac{1}{\cos(\pi/2m)}\left[\cos\left(\frac{\pi}{2m}\right)\sum_{l=1}^{\lfloor d/2\rfloor}\cos\left(\frac{2\pi l}{d}\xi\right)+\sin\left(\frac{\pi}{2m}\right)\sum_{l=1}^{\lfloor d/2\rfloor}\sin\left(\frac{2\pi l}{d}\xi\right)\right]
\end{eqnarray}
where we have denoted $\xi=n+1/2m$. Using the Euler representations of the
cosine and sine functions the above two sums can be easily computed and they read
\begin{equation}
\sum_{l=1}^{\lfloor d/2\rfloor}\cos\left(\frac{2\pi l}{d}\xi\right)=\frac{1}{2}\left[\frac{\sin(\pi\xi)}{\sin(\pi\xi/d)}-1\right]
\end{equation}
and
\begin{equation}
\sum_{l=1}^{\lfloor d/2\rfloor}\sin\left(\frac{2\pi l}{d}\xi\right)=
\frac{1}{2}\left[\cot\left(\frac{\pi\xi}{d}\right)-\frac{\cos(\pi\xi)}{\sin(\pi\xi/d)}\right].
\end{equation}
Introducing them into Eq. (\ref{Sonda}) and with the aid of some trigonometric formulas, one obtains
\begin{eqnarray}\label{rownanie}
\sum_{l\in I}c_l\omega^{nl}&=&\frac{1}{2}\left\{\frac{\sin(\pi \xi)}{\sin(\pi\xi/d)}-1+
\tan\left(\frac{\pi}{2m}\right)\left[\cot\left(\frac{\pi\xi}{d}\right)-\frac{\cos(\pi\xi)}{\sin(\pi \xi/d)}\right]\right\}\nonumber\\
&=&\frac{1}{2}\left\{\tan\left(\frac{\pi}{2m}\right)\cot\left[\frac{\pi}{d}\left(n+\frac{1}{2m}\right)\right]-1\right\}.
\end{eqnarray}
By replacing $n$ with $-(n+1)$ in the above formula we then arrive at the expression for the sum in Eq. (\ref{betan}), that is, 
\begin{equation}\label{rownanie2}
\sum_{l\in I}c_l\omega^{-(n+1)l}=-\frac{1}{2}\left\{\tan\left(\frac{\pi}{2m}\right)\cot\left[\frac{\pi}{d}\left(n+1-\frac{1}{2m}\right)\right]+1\right\}.
\end{equation}
Finally, setting $n=\lfloor d/2\rfloor=(d-1)/2$ in Eq. (\ref{rownanie}) one obtains a formula for $S$:
\begin{equation}\label{rownanie3}
S=\frac{1}{2}\left\{1-\tan\left(\frac{\pi}{2m}\right)\cot\left[\frac{\pi}{d}\left(\left\lfloor \frac{d}{2} \right\rfloor+\frac{1}{m}\right)\right]\right\}.
%S=\frac{1}{2}\left\{1-\tan\left(\frac{\pi}{2m}\right)\cot\left[\frac{\pi}{2d}\left(d-1+\frac{1}{m}\right)\right]\right\}.
\end{equation}
Substituting Eqs. (\ref{rownanie}), (\ref{rownanie2}), and (\ref{rownanie3})
into Eqs. (\ref{alphan}) and (\ref{betan}), we eventually obtain the coefficients $\alpha_n$
and $\beta_n$ in the following form
\begin{equation}\label{alphan2}
\alpha_n=\frac{1}{2d}\tan\left(\frac{\pi}{2m}\right)\left\{\cot\left[\frac{\pi}{d}\left(n+\frac{1}{2m}\right)\right]-\cot\left[\frac{\pi}{d}\left(\left\lfloor\frac{d}{2}\right\rfloor+\frac{1}{2m}\right)\right]\right\}
\end{equation}
and
\begin{equation}\label{betan2}
\beta_n=\frac{1}{2d}\tan\left(\frac{\pi}{2m}\right)\left\{\cot\left[\frac{\pi}{d}\left(n+1-\frac{1}{2m}\right)\right]+\cot\left[\frac{\pi}{d}\left(\left\lfloor\frac{d}{2}\right\rfloor+\frac{1}{2m}\right)\right]\right\}.
\end{equation}
with $n=1,\ldots,\lfloor d/2\rfloor$. As in the main text, the coefficients can be expressed using the function $g(x) := \cot(\frac{\pi}{d}(x+ \frac{1}{2m}))$.

\paragraph{Even $d$.} Clearly, in the case of even $d$, one can solve the system 
(\ref{system4}) analogously. The difference is, however, that (\ref{system4}) is the same equation for $l=-d/2$ and $l=d/2$, and therefore the system consists of $d-1$ equations for $d$ variables. A non-unique solution is then expected.

Denoting $I_e=\{-(d-1)/2,\ldots,-1,1,\ldots,d/2\}$ and following the same methodology 
as above with the set $I$ replaced by $I_e$ one arrives at $\alpha_n$ and $\beta_n$ given by 
\begin{equation}
\alpha_n=\frac{1}{2d}\left\{\tan\left(\frac{\pi}{2m}\right)\cot\left[\frac{\pi}{d}\left(n+\frac{1}{2m}\right)\right]-1\right\}+\frac{1}{d}S
\end{equation}
and
\begin{equation}
\beta_n=\frac{1}{2d}\left\{\tan\left(\frac{\pi}{2m}\right)\cot\left[\frac{\pi}{d}\left(n+1-\frac{1}{2m}\right)\right]+1\right\}-\frac{1}{d}S,
\end{equation}
where $S$ is given by the same formula as in Eq. (\ref{S}).
%Due to the aforementioned fact that for even $d$ the system (\ref{system3})has less equations than variables, 
Here, the quantity $S$ (or, equivalently, one of the variables $\alpha_n$ or $\beta_n$) cannot be uniquely determined. We fix it in such a way that the resulting $\alpha_n$ and $\beta_n$ are given by the same formulas as those in the odd $d$ case, that is, 
\begin{equation}\label{Seven}
S=\frac{1}{2}\left\{1-\tan\left(\frac{\pi}{2m}\right)\cot\left[\frac{\pi}{d}\left(
\left\lfloor\frac{d}{2}\right\rfloor+\frac{1}{2m}\right)\right]\right\}.
\end{equation}

As a consequence the coefficients $\alpha_n$
and $\beta_n$ are given by Eqs. (\ref{alphan2}) and (\ref{betan2}), both in the odd and even $d$ cases.

It is finally worth mentioning that the values of the two Bell expressions---in terms of probabilities $I_{d,m}$ and in terms of generalized correlators $\widetilde{I}_{d,m}$ ---are related in the following way:
\begin{equation}\label{relation}
\widetilde{I}_{d,m} = d I_{d,m} -2m S,
%
%\left(1-\tan\frac{\pi}{2m}\cot\left[\frac{\pi}{d}
%\left(\left\lfloor\frac{d}{2}\right\rfloor+\frac{1}{2m}\right)\right ] \right).
\end{equation}
where $S$ is given by equation (\ref{rownanie3}).

\paragraph{Special cases.} Let us now consider two special cases of $d=2$ and any $m$, and $m=2$ and any $d$. In the first one, 
the Bell expression in the probability form simplifies to
\begin{equation}
I_{2,m}=\alpha_0\mathbbm{P}_0-\beta_{0}\mathbbm{Q}_0
\end{equation} 
where
\begin{equation}
\mathbbm{P}_0=\sum_{i=1}^m[P(A_i=B_i)+P(B_i=A_{i+1})],\qquad 
\mathbbm{Q}_0=\sum_{i=1}^m[P(A_i=B_i-1)+P(B_i=A_{i+1}-1)]
\end{equation}
and
\begin{equation}
\alpha_0=\frac{1}{2\cos(\pi/2m)},\qquad \beta_0=0.
\end{equation}

Moreover, there is a unique coefficient $a_1$ and it simplifies to $1/[2\cos(\pi/2m)]$, so that in the correlator form our Bell expression for $d=2$ becomes
\begin{equation}
\widetilde{I}_{2,m}=\frac{1}{2\cos(\pi/2m)}\left[\langle A_1B_1\rangle-\langle A_1B_m\rangle+\sum_{i=2}^m\left(\langle A_iB_i\rangle+\langle A_iB_{i-1}\rangle\right)\right],
\end{equation}
and Theorems \ref{theoclass}, \ref{theoquantum} and \ref{theons} from the main text give $\widetilde{C}_b=(m-1)/\cos[\pi/2m]$, $\widetilde{Q}_b=m$, and $\widetilde{NS}_b=m/\cos[\pi/2m]$, respectively. This is the well-known chained Bell inequality \cite{chained}, which was recently used in Ref. \cite{supic16} to self-test the maximally entangled state of two qubits and the corresponding measurements. 

In the second case, i.e., that of $m=2$ and any $d$, the Bell expression $I_{d,2}$
in the probability form is given by Eq. 
\begin{equation}\label{Bellproba_app}
I_{d,2}:=\sum_{k=0}^{\left \lfloor d/2\right \rfloor-1}\left(
\alpha_k \mathbbm{P}_k -\beta_k \mathbbm{Q}_k \right ),
\end{equation}
with the expressions $\mathbbm{P}_k$ and $\mathbbm{Q}_k$ simplifying to
\begin{equation}
\mathbbm{P}_k=P(A_1=B_1+k)+P(B_1=A_{2}+k)+P(A_2=B_2+k)+P(B_2=A_{1}+k+1)
\end{equation}
and
\begin{equation}
\mathbbm{Q}_k=P(A_1=B_1-k-1)+P(B_1=A_{2}-k-1)+P(A_2=B_2-k-1)+P(B_2=A_{1}-k),
\end{equation}
where we have exploited the convention that $A_3=A_1+1$. Then, the coefficients $\alpha_k$
and $\beta_k$ are given by
\begin{equation}
\alpha_k=\frac{1}{2d}\left[g(k)+(-1)^d\tan\left(\frac{\pi}{4d}\right)\right],\qquad
\beta_k=\frac{1}{2d}\left[g\left(k + 1/2 \right)-(-1)^d\tan\left(\frac{\pi}{4d}\right)\right],\qquad
\end{equation}
with $g(k)=\cot[\pi(k+1/4)/d]$. On the other hand, in the correlator form
one obtains 
\begin{equation}
\widetilde{I}_{d,2}=\sum_{l=1}^{d-1}
\left[a_l\langle A_1^l B_1^{d-l}\rangle+a_l^*\omega^l\langle A_1^l B_2^{d-l}\rangle+a_l\langle A_2^l B_2^{d-l}\rangle+a_l^*\langle A_2^l B_1^{d-l}\rangle\right],
\end{equation}
where $a_l=\omega^{(2l-d)/8}/\sqrt{2}$. In this case Theorems \ref{theoclass}, \ref{theoquantum}, and \ref{theons} give 
\begin{equation}
\widetilde{C}_b=\frac{1}{2}\left[3\cot\left(\frac{\pi}{3d}\right)-\cot\left(\frac{3\pi}{4d}\right)\right]-2,
\end{equation}
$\widetilde{Q}_b=2(d-1)$, and $\widetilde{NS}_b=2\cot[\pi/(4d)]-2$. It should be noticed that this Bell inequality previously studied in Refs. \cite{wonminson} and \cite{devicente2015}, and, in particular in Refs. \cite{devicente2015} and \cite{lee2007} the maximal quantum violation was found using two different methods.

\section{Classical bound of the inequalities}\label{appendixclassical}
%expression (\ref{Bellproba})
We present here a detailed proof of Theorem \ref{theoclass} from the main text. Let us start with our Bell expression in the probability form $I_{d,m}$ and note that we can rewrite it as:
\begin{equation}\label{Bellproba22}
I_{d,m}:=\sum_{k=0}^{d - 1} \alpha_k \sum_{i=1}^{m}[P(A_i=B_i+k)+P(B_i=A_{i+1}+k)],
\end{equation}
with $A_{m+1}=A_1+1$. This is possible because of the form (\ref{alphan2}) and (\ref{betan2}) of coefficients $\alpha_k$ and $\beta_k$. Indeed, since $\alpha_{k} = - \beta_{d - k - 1}$, the terms of the sum which were attached to the $\beta_{k}$ coefficients can be shifted to indices $k = \lfloor d/2 \rfloor, \ldots, d - 1$ and now associated to an $\alpha_{k}$. In the odd case, we should in principle impose that the term $k = \lfloor d/2 \rfloor$ disappears, but it happens naturally since $\alpha_{\lfloor d/2 \rfloor} = 0$.

As stated in the main text, finding the classical bound of expression (\ref{Bellproba22}) reduces to computing the optimal deterministic strategy. Thus, to describe the difference between the outcomes associated to $A_x$ and $B_y$, we can assign one value $q$ such that $P(A_x = B_y + k) = \delta_{kq}$. As $q$ depends on inputs $x$ and $y$ but not all pairs of  $A_x$ and $B_y$ appear in the Bell expression, we thus define $2m$ variables $q_i \in \{0, 1 \ldots, d-1\}$ such that:
\begin{eqnarray}
A_1 - B_1 & = & q_1, \nonumber \\ 
B_1 - A_2 & = & q_2, \nonumber \\
A_2 - B_2 & = & q_3, \nonumber \\
& \vdots & \nonumber \\
A_m - B_m & = & q_{2m -1}, \nonumber \\
B_m - A_1 & = & q_{2m} + 1.
\end{eqnarray}
Due to the chained character of these equations, $q_{2m}$ must obey a superselection rule involving the other $q_i$'s, which is 
\begin{equation}
q_{2m} = -1 - \sum_{i = 1}^{2m - 1} q_i,
\end{equation}
where the sum is modulo $d$. Due to the fact that the dependence of the coefficients $\alpha_k$ on $k$ is only through the cotangent function, proving  Theorem \ref{theoclass} boils down to the following maximization problem.
%
%\begin{thm}\label{theoclasssup}%{theoclass}
\begin{repthm}{theoclass}
Let 
\begin{equation} 
\hat{\alpha}_k:=\cot\left[{\frac{\pi}{d}\left(k+\frac{1}{2m}\right)}\right], \nonumber
\end{equation}
and let 
\begin{equation}\label{Glacier}
\hat{C}_{b}:=\max_{0 \leq q_1, \ldots, q_{2m-1} < d}\left(\sum_{i=1}^{2m-1}\hat{\alpha}_{q_i} + \hat{\alpha}_{-1-\sum_{i=1}^{2m-1}q_i \mod d}\right).
\end{equation}
Then, $\hat{C}_{b} = (2m-1)\hat{\alpha}_0 + \hat{\alpha}_{d-1}.$
\end{repthm}
Notice that to recover the exact expression $\widetilde{C}_b$ from the main text, one needs to reintroduce the constant factors appearing in the definition of $\alpha_k$ and use Eq. (\ref{relation}). To prove the theorem, we first demonstrate two lemmas. Note that throughout this section, we assume that $m \geq 2$ and $d \geq 2$. Although these are not tight conditions to prove our results, they are in any case satisfied by the definition of a Bell test.

\begin{lem}\label{lemma2}
Let $g(x)=\cot[\pi(x+ \frac{1}{2m})/d]$. For all $x,y$ satisfying $0 \leq x<y<d-\frac{1}{2m}$, we have
\begin{equation}\label{equationlemma2}
(1+2mx)g(x) >(1+2my)g(y). 
\end{equation}
\end{lem}
\begin{proof}
Let us consider the function $f(z):=z\cot z$, which is strictly decreasing in the interval $0 < z < \pi$.  This can be shown for instance by noting that $f$ is holomorphic and by studying the sign of the coefficients of its Laurent series in a ball of radius $\pi$ centered at $z = 0$. Thus, for every $c\in (0, \pi)$, $f(c)>f(z)$ for all $c<z<\pi$. In particular, we can pick $c:=\frac{\pi}{2dm}(1+2mx)$ so that:
\begin{equation}
\frac{\pi}{2dm}(1+2mx) \cot\left(\frac{\pi}{2dm}(1+2mx)\right)>zf(z),
\end{equation}
for $\frac{\pi}{2dm}(1+2mx)<z<\pi$. By introducing the change of variables $z=\frac{\pi}{2dm}(1 + 2my)$, equation (\ref{equationlemma2}) follows. Note that for integer values of $x$ and $y$, namely $k$ and $l$, Lemma \ref{lemma2} becomes:
\begin{equation}
(1+2Mk)\hat{\alpha}_k > (1+2Ml) \hat{\alpha}_l, \qquad \forall 0\leq k < l < d.
\end{equation}
\end{proof}

\begin{lem}\label{lemma3}
For integer indices $k, l, p$ such that $0< k,l < d$ and $0 \leq p < d$, we have:
\begin{equation}
\hat{\alpha}_0 + \hat{\alpha}_p > \hat{\alpha}_k + \hat{\alpha}_l. 
\end{equation}
\end{lem}
\begin{proof}
Because all the alphas are ordered $\hat{\alpha}_0 > \hat{\alpha}_1 > \hat{\alpha}_2 > \cdots > \hat{\alpha}_{d-1}$, we have that $\hat{\alpha}_0 + \hat{\alpha}_p \geq \hat{\alpha}_0 + \hat{\alpha}_{d-1}$ and $\hat{\alpha}_1 + \hat{\alpha}_1 \geq \hat{\alpha}_k + \hat{\alpha}_l$. Hence, it suffices to prove that 
\begin{equation}\label{toprove1}
\hat{\alpha}_0 + \hat{\alpha}_{d-1} > 2 \hat{\alpha}_1. 
\end{equation}
Let us rewrite this inequality using the function $g$ introduced in Lemma \ref{lemma2}. To this end, we note that the symmetry of the function $\cot(x)=-\cot(-x)$ translates to $g(x)$ in the following manner : $g(x)=-g(-x-1/m)$. Thus, in order to prove (\ref{toprove1}), we need to show:
\begin{equation}\label{grelation}
g(0) > 2 g(1) + g(1-1/m).
\end{equation}
Using Lemma \ref{lemma2} twice, we can express that:
\begin{equation}
g(0) > (2m-1)g(1-1/m) > g(1-1/m) + 2(m-1)\frac{(1+2m)}{(2m-1)} g(1).
\end{equation}
To obtain the second inequality, one of the $2m -1$ terms was isolated, and Lemma \ref{lemma2} was applied only on the remaining $2(m-1)$ terms. The minimum of $2(m-1)(1+2m)/(2m-1)$ is found for $m=2$ and it is equal to $10/3$. Since $g(1)$ is positive, and $10/3 > 2$, we can conclude that $g(0) > g(1-1/m) + 2g(1)$, which is exactly relation (\ref{grelation}).
\end{proof}

\begin{proof}[Proof of Theorem \ref{theoclass}]
To demonstrate the theorem, we employ a dynamic programming procedure which allows us to rewrite Eq. (\ref{Glacier}) as a chain of maximizations, each over a single variable. Let us first define 
\begin{equation}
h(x):=\max_{0\leq y < d} \left(\hat{\alpha}_y + \hat{\alpha}_{-1-x-y}\right),
\end{equation}
where the indices are taken to be modulo $d$. As a direct consequence of Lemma \ref{lemma3}, $h(x)=\hat{\alpha}_0 + \hat{\alpha}_{-1-x}$. Indeed, the lemma implies that $\hat{\alpha}_0 + \hat{\alpha}_{-1-x} > \hat{\alpha}_y + \hat{\alpha}_{-1-x-y}$ if $y>0$ and $x\neq d- 1 - y$. For the cases where $y = 0$ or $x = d - 1 - y$, the maximum is directly attained.
This allows us to write the classical bound as:
\begin{equation}\label{classmax}
\hat{C}_b = \max_{q_1} \left(\hat{\alpha}_{q_1} + \max_{q_2} \left(\hat{\alpha}_{q_2} + \ldots + \max_{q_{2m-2}}\left(\hat{\alpha}_{q_{2m-2}} + h\left(\sum_{i=1}^{2m-2}q_{i}\right)\right)\ldots\right)\right).
\end{equation}
Using the properties of $h$, we find that 
\begin{equation}
\max_{q_{k}}\left[\hat{\alpha}_{q_{k}} + h\left(\sum_{i=1}^{k}q_{i}\right)\right] =\hat{\alpha}_0 + h\left(\sum_{i=1}^{k-1}q_i\right)
\end{equation}
for all $k$. By applying this step $2(m -1)$ times to expression (\ref{classmax}), we obtain:
\begin{equation}
\hat{C}_b = (2m - 2)\hat{\alpha}_0 + h(0) = (2m -1) \hat{\alpha}_0 + \hat{\alpha}_{-1}.
\end{equation}
\end{proof}

\section{Tsirelson bound of the inequalities}\label{appendixtsirelson}
Here, we present more details on the SOS decomposition of any Bell operator corresponding to our new Bell inequality $\widetilde{I}_{d,m}$, thus complementing the proof of Theorem \ref{theoquantum} from the main text. Concretely, we show that the identity 
\begin{equation}\label{sos_app}
\widetilde{Q}_{b}\mathbbm{1}-\mathcal{B}  = \frac{1}{2}\sum_{i=1}^{m}\sum_{k=1}^{d-1} P_{ik}^{\dagger} P_{ik} +  \frac{1}{2}\sum_{i=1}^{m - 2}\sum_{k=1}^{d-1} T_{ik}^{\dagger} T_{ik},
\end{equation}
is valid independently of the choice of $A_i^k$ and $B_i^k$. The operators are thus not specified. Here, $P_{ik} =\mathbbm{1}\ot\bar{B}_i^k - (A_i^k)^\dagger\ot\mathbbm{1}$, and 
\begin{equation}
T_{ik} =\mu_{i,k} B_{2}^{d - k} + \nu_{i,k} B_{i + 2}^{d - k} + 
\tau_{i,k} B_{i + 3}^{d - k},
\end{equation}
where the coefficients $\mu_{ik}$, $\nu_{ik}$ and $\tau_{ik}$ are given by 
\begin{eqnarray}
\mu_{i,k} & = & \frac{\omega^{(i + 1)(d - 2k)/2m}}{2 \cos (\pi/2m)} \frac{\sin (\pi/m)}{\sqrt{\sin (\pi i/m) \sin \left[\pi (i +1)/m\right]}}, \nonumber \\
\nu_{i,k} & = & -\frac{\omega^{(d - 2k)/2m}}{2 \cos (\pi/2m)} \sqrt{\frac{\sin \left[\pi (i +1)/m\right]}{\sin (\pi i/m)}}, \nonumber \\
\tau_{i,k} & = & \frac{1}{2 \cos (\pi/2m)} \sqrt{\frac{\sin (\pi i/m)}{ \sin \left[\pi (i +1)/m\right]}}=-\frac{\omega^{(d - 2k)/2m}}{4\cos^2(\pi/2m)}\nu_{ik}^{-1},
\label{coeffabc}
\end{eqnarray}
for $i = 1, \ldots, m -3$ and $k = 1, \ldots, d-1$, while for $i=m-2$ and $k=1,\ldots,d-1$ they are given by 
\begin{eqnarray}
\mu_{m-2,k} & = & - \frac{\omega^{-(d - 2k)/2m}}{2\sqrt{2} \cos (\pi/2m)\sqrt{\cos (\pi/m)}},\nonumber \\
\nu_{m-2,k} & = & - \frac{\omega^k \omega^{(d - 2k)/2m}}{2\sqrt{2} \cos (\pi/2m)\sqrt{\cos (\pi/m)}},\nonumber \\
\tau_{m-2,k} & = & \frac{\sqrt{\cos (\pi/m)}}{\sqrt{2}\cos (\pi/2m)}.
\label{coeffm2}
\end{eqnarray}

Now, in order to check the validity of the SOS decomposition (\ref{sos_app})
let us first introduce the explicit form of $P_{ik}$ into the first term
of the right-hand side of (\ref{sos_app}), which gives
\begin{equation}\label{Druid}
\sum_{i=1}^{m}\sum_{k=1}^{d-1} P_{ik}^{\dagger} P_{ik}=\widetilde{Q}_b\mathbbm{1}-2\mathcal{B}+\mathbbm{1}\ot \sum_{i=1}^{m}\sum_{k=1}^{d-1}(\bar{B}_i^{k})^{\dagger}(\bar{B}_i^{k}),
\end{equation}
where we have used the fact that the Bell operator $\mathcal{B}$ is Hermitian.

Let us then introduce the explicit form of the operators $T_{ik}$ into the last term of 
the right-hand side of (\ref{sos_app}), which, after some simple algebra, leads us to
\begin{eqnarray}\label{Gilgamesh}
\sum_{i=1}^{m - 2}\sum_{k=1}^{d-1} T_{ik}^{\dagger} T_{ik}&=&\sum_{i=1}^{m-2}
\sum_{k=1}^{d-1}\left(|\mu_{i,k}|^2+|\nu_{i,k}|^2+|\tau_{i,k}|^2\right)\mathbbm{1}\nonumber\\
&&+\sum_{k=1}^{d-1}\left[\mu_{1,k}^*\nu_{1,k}(B_2^{d-k})^{\dagger}(B_3^{d-k})+\mu_{1,k}\nu_{1,k}^*(B_3^{d-k})^{\dagger}(B_2^{d-k})\right]\nonumber\\
&&+\sum_{k=1}^{d-1}\left[\mu_{m-2,k}^*\tau_{m-2,k}(B_2^{d-k})^{\dagger}(B_1^{d-k})+\mu_{m-2,k}\tau_{m-2,k}^*(B_1^{d-k})^{\dagger}(B_2^{d-k})\right]\nonumber\\
&&+\sum_{i=1}^{m-3}\sum_{k=1}^{d-1}\left[(\mu_{i,k}^*\tau_{i,k}+\mu_{i+1,k}^{*}\nu_{i+1,k})(B_2^{d-k})^{\dagger}(B_{i+3}^{d-k})\right.\nonumber\\
&&\left.\hspace{2cm}+(\mu_{i,k}\tau_{i,k}^*+\mu_{i+1,k}\nu_{i+1,k}^*)(B_{i+3}^{d-k})^{\dagger}(B_2^{d-k})\right]\nonumber\\
&&+\sum_{i=1}^{m-2}\sum_{k=1}^{d-1}\left[\nu_{i,k}^*\tau_{i,k}(B_{i+2}^{d-k})^{\dagger}(B_{i+3}^{d-k})+\nu_{i,k}\tau_{i,k}^*(B_{i+3}^{d-k})^{\dagger}(B_{i+2}^{d-k})\right].
\end{eqnarray}
Now, it follows from Eqs. (\ref{coeffabc}) and (\ref{coeffm2}) that 
$\mu_{i,k}^*\tau_{i,k}+\mu_{i+1,k}^{*}\nu_{i+1,k}=0$ for $i=1,\ldots,m-3$ and $k=1,\ldots,d-1$, which means that the fourth and fifth lines in the above vanish. Then, 
one notices that $\mu_{1,k}^*\nu_{1,k}=\mu_{m-2,k}\tau_{m-2,k}^*=\nu_{i,k}^*\tau_{i,k}=-a_k^2$ for $i=1,\ldots,m-3$ and $k=1,\ldots,d-1$, and $\nu_{m-2,k}\tau_{m-2,k}^*=-\omega^k (a_k^*)^2$ for $k=1,\ldots,d-1$, where, as before, $a_k=\omega^{-(d-2k)/4m}/[2\cos(\pi/2m)]$. Therefore, the remaining terms on the right-hand side of Eq. (\ref{Gilgamesh})
can be wrapped up as
\begin{eqnarray}\label{Gilgamesh2}
\sum_{i=1}^{m - 2}\sum_{k=1}^{d-1} T_{ik}^{\dagger} T_{ik}&=&\sum_{i=1}^{m-2}
\sum_{k=1}^{d-1}\left(|\mu_{ik}|^2+|\nu_{ik}|^2+|\tau_{ik}|^2\right)\mathbbm{1}\nonumber\\
&&-\sum_{i=1}^{m-1}\sum_{k=1}^{d-1}\left[a_k^2(B_i^{d-k})^{\dagger}(B_{i+1}^{d-k})+(a_k^*)^2(B_{i+1}^{d-k})^{\dagger}(B_i^{d-k})\right]\nonumber\\
&&-\sum_{k=1}^{d-1}\left[\omega^k (a_k^*)^2(B_1^{d-k})^{\dagger}(B_{m}^{d-k})+\omega^{-k}a_k^2(B_{m}^{d-k})^{\dagger}(B_1^{d-k})\right].
\end{eqnarray}
By substituting Eqs. (\ref{Druid}) and (\ref{Gilgamesh2}) into
Eq. (\ref{sos_app}) and exploiting the explicit form of the operators 
$\bar{B}_i^k$, one obtains
\begin{eqnarray}
\frac{1}{2}\sum_{i=1}^{m}\sum_{k=1}^{d-1} P_{ik}^{\dagger} P_{ik}+
\frac{1}{2}\sum_{i=1}^{m - 2}\sum_{k=1}^{d-1} T_{ik}^{\dagger} T_{ik}&=&\frac{1}{2}\widetilde{Q}_b
\mathbbm{1}-\mathcal{B}\nonumber\\
&&+\sum_{k=1}^{d-1}\left[m|a_k|^2+\frac{1}{2}\sum_{i=1}^{m-2}\left(|\mu_{i,k}|^{2} + |\nu_{i,k}|^{2} + |\tau_{i,k}|^{2} \right) \right]\mathbbm{1}.\nonumber\\
\end{eqnarray}
It is easy to finally realize that the last two terms in the above formula amount to
$(1/2)\widetilde{Q}_b=(1/2)m(d-1)$, which completes the proof.

\section{No-signalling bound of the inequalities}\label{appendixns}
Here, we present details on the proof of Theorem \ref{theons} from the main text. As for the section on the classical bound of our inequalities, we start from the Bell expression written as:
\begin{equation}\label{Bellproba2}
I_{d,m}:=\sum_{k=0}^{d - 1} \alpha_k \sum_{i=1}^{m}[P(A_i=B_i+k)+P(B_i=A_{i+1}+k)],
\end{equation}
with $A_{m+1}=A_1+1$. Following considerations from that section, %Section \ref{appendixclassical}, 
it is clear that the coefficient $\alpha_0$ is the largest of the sum. Thus, the algebraic bound of $I_{d,m}$ is then $2m\alpha_0$. To complete the proof, we provide a no-signalling behaviour that reaches this bound. Let us recall the no-signalling conditions for a probability distribution:
\begin{eqnarray}\label{nosigncond}
\sum_{b} P(A_x = a, B_y = b) = \sum_{b} P(A_x = a, B_{y'}= b) \qquad \forall a, x, y, y' \nonumber \\
\sum_{a} P(A_x = a, B_y = b) = \sum_{a} P(A_{x'} = a, B_y = b) \qquad \forall b, y, x, x',
\end{eqnarray}
which express that the marginals on Alice's side do not depend on Bob's input, and conversely. The behaviour that we present is the following.  For inputs $x$ and $y$ such that $x = y$ or $x = y+1$:
\begin{equation}
P(A_y = a, B_y = b) = P(A_{y+1} = a, B_y = b) =  \left\{
\begin{array}{cll}
1/d & \mathrm{if} & a = b \\
0 & \mathrm{if} & a \neq b.
\end{array}
\right.
\end{equation}
There is a special case for $x = 1$ and $y = m$:
\begin{equation}
P(A_1 = a, B_m = b) = \left\{
\begin{array}{cll}
1/d & \mathrm{if} & a = b - 1 \\
0 & \mathrm{if} & a\neq b - 1,
\end{array}
\right.
\end{equation}
where the addition is modulo $d$. For all the other input combinations (i.e. the ones not appearing in the inequalities), we have:
\begin{equation}
P(A_x = a, B_y = b) = 1/d^2 \qquad \forall a,b.
\end{equation}
One can easily verify that this distribution satisfies conditions (\ref{nosigncond}). To obtain the expression from Theorem \ref{theons}, it suffices to write explicitly $2m\alpha_0$ and to use relation (\ref{relation}).

\section{Scaling of the bounds}\label{appendixscalings}
Here, we study the asymptotic behaviour of the bounds of our Bell expressions for large numbers of inputs $m$ and outputs $d$. This can be of interest when studying applications in device-independent protocols, for instance. We also show that for any values of $m$ and $d$, the classical bound is strictly smaller than the quantum bound, which is strictly smaller than the no-signalling bound. This ensures in particular that the Bell inequality is never trivial. 

Let us start with the quantity:
\begin{equation}\label{quantumclassical}
\frac{\widetilde{Q}_b}{\widetilde{C}_b} = \frac{2m(d - 1)}{\tan\left(\frac{\pi}{2m}\right)\left[(2m - 1)\cot\left(\frac{\pi}{2dm}\right) - \cot\left(\frac{\pi}{d}(1 - \frac{1}{2m})\right)\right] - 2m}
\end{equation}
which is the ratio between the quantum and classical bounds. We also consider the ratio between the no-signalling and quantum bounds, which is:
\begin{equation}\label{quantumns}
\frac{\widetilde{NS}_b}{\widetilde{Q}_b} = \frac{\tan\left(\frac{\pi}{2m}\right)\cot\left(\frac{\pi}{2dm}\right) - 1}{d - 1}.
\end{equation}
To observe the behaviour of these quantities for high number of inputs $m$ and outputs $d$, we can use the Taylor series expansion in two variables, $1/m$ and $1/d$, and keep the dominant terms. We obtain:
\begin{eqnarray}\label{quantumclassicalinf}
\frac{\widetilde{Q}_b}{\widetilde{C}_b}  &=& 1 + \frac{1}{2m} - \frac{\pi^2 - 6}{12 m^2} + \cdots \\
\frac{\widetilde{NS}_b}{\widetilde{Q}_b} &=& 1 + \frac{ \pi^2/12 - \pi^2/12d^2}{m^2} + \cdots
\end{eqnarray}
Thus, when the parameters $m$ and $d$ are of the same order and both very large, i.e. $m = \Theta(d)$, both ratios tend to $1$. It is interesting to consider how fast the bounds tend towards each other: since the ratio between the no-signalling and quantum bounds lacks a term in $1/m$, it is clear that the quantum bound approaches the no-signalling bound faster than the classical bound approaches the quantum bound.

If we fix the number of outputs $d$ and consider the limit of a large number of inputs $m$, the ratios still tend to $1$. However, if we fix $m$ and considers the limit of large $d$, both ratios tend to constants which are a bit bigger than $1$. They are :
\begin{eqnarray}\label{quantumclassicalinf}
\lim_{d \rightarrow \infty} \widetilde{Q}_b/\widetilde{C}_b  &=& \frac{\left(2m - 1\right) \pi \cot\left(\pi/2m\right)}{4m (m-1)}  \\
\lim_{d \rightarrow \infty}\widetilde{NS}_b/\widetilde{Q}_b &=&\frac{2}{\pi}m\tan\left(\frac{\pi}{2m}\right).
\end{eqnarray}
It is worth mentioning that both functions of $m$ appearing on the right-hand sides of 
the above formulas attain their maxima for $m=2$ which are $4/\pi$ and $3\pi/8$, respectively. To give the reader more insight, we present in Tables \ref{table1} and \ref{table2} the numerical values of these ratios for low values of $m$ and $d$.

Now, let us show that these ratios are strictly larger than $1$ for any value of $m$ and $d$ consistent with a Bell scenario. 
\begin{lem}
For any $m,d \geq 2$, the quantum bound of $\widetilde{I}_{d,m}$ is strictly larger than the classical one, that is,
\begin{equation}\label{lemmaclasseq}
\widetilde{Q}_b/\widetilde{C}_b > 1.
\end{equation}
\end{lem}
\begin{proof}
We prove that $\widetilde{Q}_b - \widetilde{C}_b > 0$, which is equivalent  to (\ref{lemmaclasseq}) since both bounds are larger than $0$. This inequality can be written as:
\begin{equation}
2md\cot\left(\frac{\pi}{2m}\right) - 2m\cot\left(\frac{\pi}{2dm}\right) + \cot\left(\frac{\pi}{2dm}\right) + \cot\left(\frac{\pi}{d} \left(1 - \frac{1}{2m}\right)\right) > 0.
\end{equation}
If we define $a = 1/d$ and $x = \pi/2m$, it becomes:
\begin{equation}
ax\cot(a(\pi - x)) + a(x - \pi)\cot(ax) + \pi\cot(x) > 0,
\end{equation}
 for $0<a\leq 1/2$ and $0 < x \leq \pi/4$. Since the first term is positive for these intervals, it suffices to show that
\begin{equation}
u(a,x):= a(x - \pi)\cot(ax) + \pi\cot(x) > 0.
\end{equation}
Clearly, $ u(a,x) \geq \min_{a}(u(a,x))$. This minimum corresponds to the limit $a \rightarrow 0$, since the derivative  $\partial u(a,x)/\partial a$ of  $u(a,x)$ with respect to $a$ is strictly positive on the considered intervals of $a$ and $x$. 
Indeed, it holds that
\begin{equation}
\frac{\partial u(a,x)}{\partial a}=(x - \pi) \cot(ax) - \frac{ax(x - \pi)}{\sin^2(ax)},
\end{equation}
which can be rewritten as 
\begin{equation}\label{Cidra}
\frac{\partial u(a,x)}{\partial a}=\frac{\pi-x}{2\sin^2(ax)}\left[2ax - \sin(2ax)\right].
\end{equation}
Now, due to the fact that $y>\sin y$ for $0<y\leq \pi/8$, one has that 
$2ax>\sin(2ax)$ for $0<a\leq 1/2$ and $0<x\leq \pi/4$, 
and therefore the right-hand side of Eq. (\ref{Cidra}) is strictly 
positive within the above intervals.

Now, computing the limit of $u(a,x)$ when $a \rightarrow 0$, one obtains
\begin{equation}
\lim_{a \rightarrow 0} u(a,x) = 1 - \frac{\pi}{x} + \pi \cot(x).
\end{equation}
It can be verified straightforwardly that this expression is strictly positive in the interval $0 < x \leq \pi/4$, by comparing the two functions $\pi \cot(x)$ and $\frac{\pi}{x} - 1$, and noticing that the former upper bounds the latter in the interval $0 < x \leq \pi/4$. Indeed, at $x =  \pi/4$, we have that $\pi \cot(\pi/4) > 3$, and in this interval, both their derivatives are negative, with the derivative of the first function smaller than the derivative of the second one. Thus, $u(a,x) > 0$.
\end{proof}

\begin{lem}
For any $m,d \geq 2$, the no-signalling bound of $\tilde{I}_{d,m}$ is strictly larger than the quantum one, that is, 
\begin{equation}\label{lemmanseq}
\widetilde{NS}_b/\widetilde{Q}_b  > 1.
\end{equation}
\end{lem}
\begin{proof}
Writing the inequality explicitely as in (\ref{quantumns}),
it follows that it is enough to show that 
$\tan(\pi/2m)\cot(\pi/2dm)>d$. Let us prove a slightly simpler inequality:
\begin{equation}\label{simplereq}
\tan(\pi/2m)>d\tan(\pi/2dm).
\end{equation}

To this end, we show that $\tan(ax)> a\tan(x)$ for any $0 < x \leq \pi/2a$ and any integer $a\geq 2$. 
We notice that for $x=0$, $\tan(0)=a\tan(0)$, 
and that $[\tan(ax)]'\geq [a\tan(x)]'\geq 0$, meaning that both $\tan(ax)$ and $a\tan(x)$
are monotonically increasing functions and that the former grows faster than the latter.
The inequality for the derivatives holds true because $\cos(x)$ is a monotonically decreasing function for $0 \leq x \leq \pi/2a$ which implies that $\cos(x)\geq \cos(ax)$.

To complete the proof we note that $\tan(\pi/2m)=\tan[d (\pi/2dm)]$ % and since $\pi/2dm<\pi/2m$ for any $d\geq 2$,
and using $x = \pi/2dm$ and $a = d$, one can exploit the above inequality to obtain Eq. (\ref{simplereq}). %$\tan(\pi/2m)> d\tan(\pi/2dm)$. 
This finally implies Eq. (\ref{lemmanseq}). 
\end{proof}

\begin{table}[ht]
\centering
\begin{tabular}{c|ccccc}
 \backslashbox{d}{m} & 2 & 3 & 4 & 5 & 6\\\hline
2 & 1.414 & 1.299 & 1.232 & 1.189 & 1.159 \\
3 & 1.291 & 1.214 & 1.167 & 1.137 & 1.116 \\
4 & 1.252 & 1.186 & 1.146 & 1.120 & 1.102 \\
5 & 1.233 & 1.173 & 1.136 & 1.112 & 1.095 \\
6 & 1.222 & 1.165 & 1.130 & 1.107 & 1.091 
\end{tabular}
\caption{Numerical values of the ratio $\widetilde{Q}_b/\widetilde{C}_b$ for low number of inputs $m$ and outputs $d$. For $m= d = 2$, one recovers the well-known CHSH $\sqrt{2}$ ratio.}
\label{table1}
\end{table}

\begin{table}[ht]
\centering
\begin{tabular}{c|ccccc}
 \backslashbox{d}{m} & 2 & 3 & 4 & 5 & 6\\\hline
2 & 1.414 & 1.155 & 1.082 & 1.051 & 1.035 \\
3 & 1.366 & 1.137 & 1.073 & 1.046 & 1.031 \\
4 & 1.342 & 1.128 & 1.069 & 1.043 & 1.029 \\
5 & 1.328 & 1.123 & 1.066 & 1.041 & 1.028 \\
6 & 1.319 & 1.120 & 1.064 & 1.040 & 1.027 
\end{tabular}
\caption{Numerical values of the ratio $\widetilde{NS}_b/\widetilde{Q}_b$ for low number of inputs $m$ and outputs $d$. For $m= d = 2$, one recovers the well-known CHSH $\sqrt{2}$ ratio.}
\label{table2}
\end{table}

\section{Device-independent quantum key distribution}\label{diqkd}

We clarify here our claim that using the maximally entangled state in DI quantum key distribution can lead to better key generation rates, and illustrate it with a simple example. This example shows a case where our inequalities can be more useful than the CGLMP inequalities, despite their lower resistance to noise. We leave out a more general analysis of the key generation rates to a work focused on DIQKD.

We consider the class of protocols studied for instance in \cite{masanes11}. As explained there, the first step of the protocol consists of Alice and Bob making measurements on the copies of bipartite quantum systems that are distributed to them. For a number of rounds $N$, their inputs are set to fixed values, $x = x^{*}$ and $y = y^{*}$, and the outcomes they obtain constitute their two versions of the raw key $\vec{a} = (a_1, a_2, \cdots, a_N)$ and $\vec{b} = (b_1, b_2, \cdots, b_N)$. For a small number of rounds, which can be taken for instance as $N_{\text{est}} = \sqrt{N}$, the inputs are chosen uniformly at random, and the outputs are used to estimate the degree of nonlocality of their correlations, for instance through the violation of a Bell inequality. Note that the type of the rounds is not predetermined, so that an eavesdropper cannot know if a given round will be a key generation round or a Bell inequality violation round. The next steps of the protocol are classical, with an error-correcting stage, where Alice publishes a message about $\vec{a}$ which is used by Bob to correct his errors so that they possess the same secret key at the end.

As stated in \cite{masanes11}, the length of this secret key is lower bounded by $H_{\text{min}} ( \vec{a} | E) - N_{\text{pub}}$, i.e. the min-entropy of Alice's raw key $\vec{a}$ conditioned on an eavesdropper's information, minus the length of the message published by Alice in the error-correcting step. The idea behind our claim is that if Alice and Bob have perfect correlations, the term $N_{\text{pub}}$ amounts to $0$ and leads to a longer secret key. For simplicity, we work in an ideal case (no finite size corrections) and the quantity we study is the asymptotic key generation rate $K$, which can be lower bounded by
\begin{equation}
K \geq H_{\text{min}}(A_{x^{*}}|E) - H(A_{x^{*}}|B_{y^{*}}).
\end{equation}
The first term corresponds to the guessing probability since $H_{\text{min}}(A_{x^{*}}|E) = - \text{log}_d P_{\text{guess}}(a|x^{*})$ and can be bounded numerically using the  Navascu\'es-Pironio-Ac\'in (NPA) hierarchy \cite{npa2007}, based on the violation of a Bell inequality. The second term is the conditional Shannon entropy defined as $H(A_{x^{*}}|B_{y^{*}}) = \sum_{a,b} - P(ab|x^{*}y^{*}) \text{log}_d P(a|bx^{*}y^{*})$. Thus, the more the outcomes of Alice and Bob are correlated for the settings $x^{*}$ and $y^{*}$, the smaller this second term is. 

Let us consider an example, for the simple scenario of $m=2$ and $d=3$. Alice and Bob test the violation of a Bell inequality (CGLMP or ours, $I_{3,2}$) to certify the security of their outcomes. The guessing probability in both cases is found to be equal to $1/3$ at the maximal violation. To generate the key, Alice uses her first setting $A_1$ and Bob a third measurement $B_3$ which is chosen to be the same as $A_1$ (defined in expression (\ref{cglmpmeas})). For our inequality, in the optimal case, this leads to  $H(A_1|B_3) = 0$, since the state is the maximally entangled state and the correlations are thus perfect. For CGLMP, $H(A_1|B_3) = 0.0618$ since the optimal state is $|\psi_{\gamma}\rangle = \frac{|00\rangle + \gamma|11\rangle + |22\rangle}{\sqrt{2 + \gamma^2 }}$, with $\gamma = (\sqrt{11} - \sqrt{3})/2 $ as found in \cite{acindurt, swaptest1}. A numerical optimization on the measurement $B_3$ shows that the best choice to minimize $H(A_1|B_3)$  is indeed to set $B_3$ to be the same as $A_1$. Thus, in the ideal case where the maximal violation is observed, we have
\begin{eqnarray}
K_{I_{3,2}} \geq 1, \\
K_{\text{CGLMP}} \geq 0.9382,
\end{eqnarray}
i.e. our inequality guarantees a key rate of $1$ trit, while CGLMP guarantees a key rate of $0.9382$ trits. 

Let us now consider the effect of white noise on this example. The noise is described by parameter $\eta$, and affects the optimal state $|\psi\rangle$ as:
\begin{equation}
\rho' = (1-\eta) |\psi\rangle\langle\psi| + \eta \frac{\mathbb{I}}{d^2},
\end{equation}
which leads to a non-maximal violation of the Bell inequality. The results are shown in Figure \ref{qkd}. Up until a noise level of $\eta \approx 0.0428$, i.e. $4.3$ percent, our inequality leads to a higher key rate than CGLMP. Around $\eta \approx 0.102$, the key rate has fallen to $0$ for both inequalities.

\begin{figure}
\includegraphics[width=0.7\textwidth]{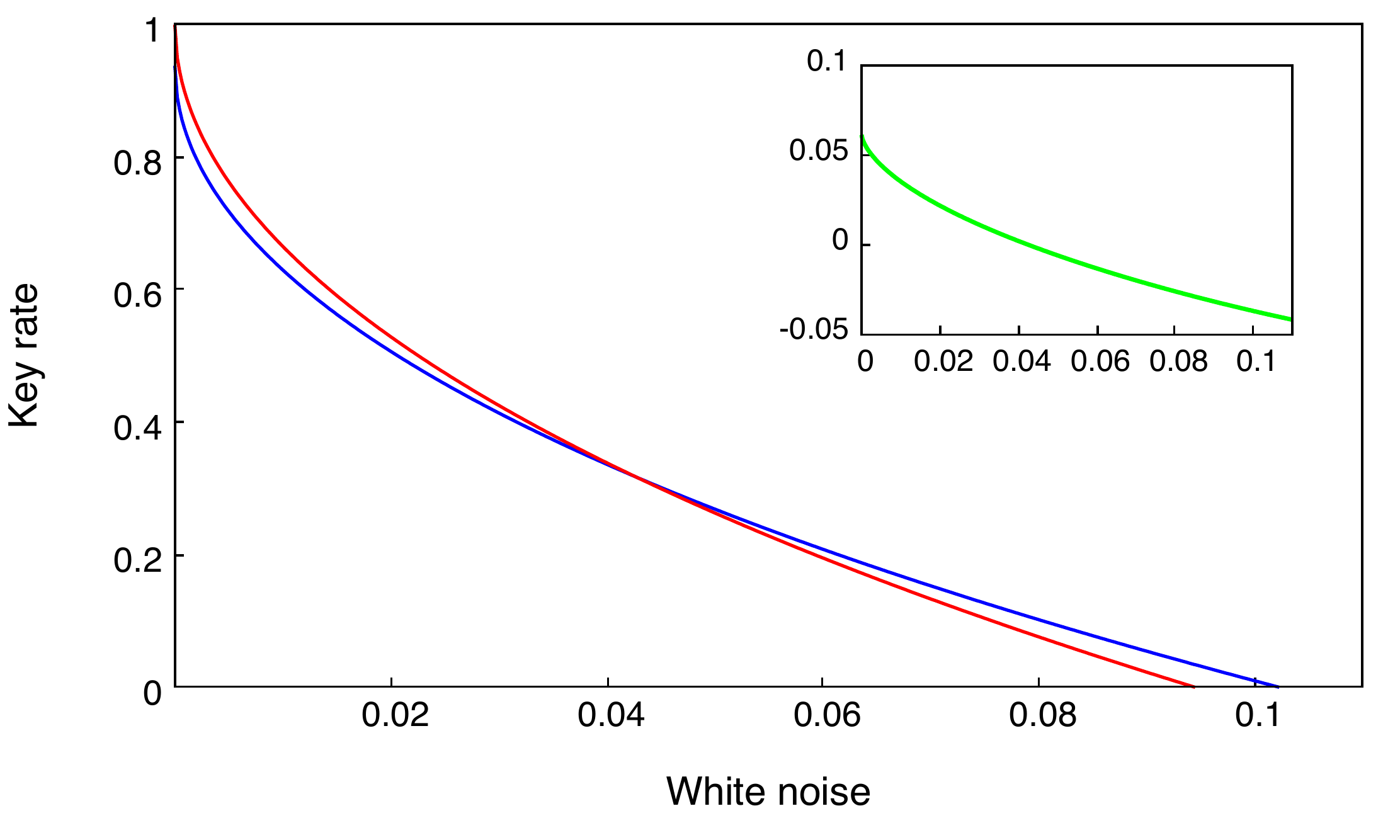}
\caption{Asymptotic key rate $K$ as a function of the white noise $\eta$. The red curve corresponds to the key rate certified with our inequality $I_{3,2}$, while the blue curve corresponds to key rate with CGLMP. On the top right, the difference between the two key rates is plotted as a function of the white noise $\eta$.\label{qkd}}
\end{figure}

Note that our bounds on the guessing probability were obtained numerically, thus this method is limited to simple scenarios. Proving such bounds analytically remains an open question, both for CGLMP and for our inequalities. Nevertheless, we can make some conjectures about the general case. 

In particular, when the maximal violation is observed without any noise, we expect that the eavesdropper does not possess any information, i.e. $H_{\text{min}}(A_{x^{*}}|E) = 1$. This conjecture allows us to connect the key rate to the quantum mutual information $I(A:B)$:
\begin{equation}
K^{\eta = 0} \geq H_{\text{min}}(A_{x^{*}}|E) - H(A_{x^{*}}|B_{y^{*}}) = H(A_{x^{*}}) - H(A_{x^{*}}|B_{y^{*}}) \equiv I(A_{x^{*}}:B_{y^{*}}).
\end{equation}
One can easily compute the mutual information for the case when projective measurements are applied on a bipartite pure state $|\psi_{AB}\rangle$. It is straightforward to see that the mutual information is upper bounded by the entanglement entropy of the state, $I(A:B) \leq E(|\psi_{AB}\rangle)$. For a state $\rho_{AB} = |\psi_{AB}\rangle \langle \psi_{AB}|$, the entropy of entanglement \cite{bruss} is defined as 
\begin{equation}
E(|\psi_{AB}\rangle) = - \text{Tr}(\rho_A \text{log} \rho_A) = - \text{Tr}(\rho_B \text{log} \rho_B),
\end{equation}
with the reduced density matrices $\rho_A = \text{Tr}_B (\rho_{AB})$ and $\rho_B = \text{Tr}_A (\rho_{AB})$ (here we use logarithm to base $d$).  The bound is tight, i.e. $I(A:B) = E(|\psi_{AB}\rangle)$, when the measurements are performed in the Schmidt basis of the state, which corresponds to the best possible choice of measurements $x^{*}, y^{*}$ to generate a secret key, given that state. Note that implementing these Schmidt basis measurements in the protocol may not be possible, depending on the Bell inequality used and its own optimal measurements.

In \cite{zohrengill}, the authors investigated numerically the states that maximally violate the CGLMP inequalities, and they found that their entanglement entropy decreases as a function of $d$. On the other hand, the entanglement entropy of the maximally entangled state is equal to 1 and independent of the dimension. Since this quantity upper bounds the mutual information, these results indicate that the key rate for $\eta = 0$ would decrease monotonically with $d$ for the CGLMP states, while our key rate would remain equal to 1. In conclusion, we can conjecture in the noiseless case that the advantage of our inequality over CGLMP grows with the dimension of the systems used for DIQKD.

\section{Structure of the set of quantum correlations}\label{appendixstructure}

We discuss in this section an aspect of our results that is linked to the fundamental question of the study of the set of quantum correlations. In particular, our results allow us to gain insight into the structure of the boundary of this set. Indeed, a feature of our inequalities worth highlighting is that their Tsirelson bound corresponds to the bound obtained using the NPA hierarchy at the first level ${\mathcal Q}_1$. This is a rare property, which has been previously observed only for XOR games (see, e.g., \cite{Wehner}) and follows from our SOS decomposition  (see Eq. (\ref{sos_app})). Indeed, the degree of an optimal SOS decomposition for a Bell operator is directly linked to the level of the NPA hierarchy at which the quantum bound is obtained \cite{PNA}. An SOS of degree one, as in our case, corresponds to the first level $\mathcal{Q}_1$. 

This means that the boundaries of the sets $\mathcal{Q}$ and $\mathcal{Q}_{1}$ intersect at the maximal violation of our inequalities.
This observation along with the results of Ref. \cite{devicente2015}
seem to suggest that the boundaries of $\mathcal{Q}$ and $\mathcal{Q}_{1}$ intersect at points that correspond to the maximal violation of Bell inequalities attained by maximally entangled states. Notice, however, that the opposite implication is not true. That is, there exist Bell inequalities whose maximal violation by the maximally entangled state does not correspond to the intersection of $\mathcal{Q}$ and $\mathcal{Q}_1$ \cite{Liang}. The above property, if proven in general, could be used to characterize $\mathcal{Q}_{1}$.

\end{document}